\documentclass{article}

\usepackage[affil-it]{authblk}
\usepackage[usenames,dvipsnames]{xcolor}
\usepackage{amsfonts}
\usepackage{amsmath,amsthm,amssymb,dsfont}
\usepackage{enumerate}
\usepackage[english]{babel}
\usepackage{graphicx}	
\usepackage[caption=false]{subfig}
\usepackage[margin=3cm]{geometry}
\usepackage{url}
\usepackage{todonotes}
\usepackage{bbm}

\usepackage{tikz}
\usetikzlibrary{chains}
\usetikzlibrary{fit}
\usepackage{pgflibraryarrows}		
\usepackage{pgflibrarysnakes}		

\usepackage{makecell}

\usepackage{epsfig}
\usetikzlibrary{shapes.symbols,patterns} 
\usepackage{pgfplots}

\usepackage{hyperref}
\hypersetup{colorlinks=true,citecolor=blue,linkcolor=blue,filecolor=blue,urlcolor=blue,breaklinks=true}

\usepackage{nicefrac}
\usepackage{mathtools}

\theoremstyle{plain}
\newtheorem{theorem}{Theorem}[section]
\newtheorem{lemma}[theorem]{Lemma}

\newtheorem{corollary}[theorem]{Corollary}
\newtheorem{proposition}[theorem]{Proposition}

\theoremstyle{definition}

\newtheorem{remark}[theorem]{Remark}
\newtheorem{example}[theorem]{Example}

\newcommand*{\PP}{\mathbb{P}}
\newcommand*{\E}{\mathbb{E}}
\newcommand*{\ee}{\mathrm{e}}

\newcommand*{\cA}{\mathcal{A}}

\newcommand*{\cD}{\mathcal{D}}

\newcommand*{\cF}{\mathcal{F}}
\newcommand*{\cG}{\mathcal{G}}
\newcommand*{\cH}{\mathcal{H}}

\newcommand*{\cM}{\mathcal{M}}

\newcommand*{\cS}{\mathcal{S}}
\newcommand*{\cT}{\mathcal{T}}

\newcommand*{\cX}{\mathcal{X}}
\newcommand*{\cY}{\mathcal{Y}}

\newcommand*{\X}{\mathbb{X}}

\newcommand*{\N}{\mathbb{N}}

\newcommand*{\R}{\mathbb{R}}

\newcommand*{\C}{\mathbb{C}}

\newcommand*{\GL}{\mathrm{GL}}
\newcommand*{\SL}{\mathrm{SL}}

\newcommand*{\eps}{\varepsilon}
\newcommand*{\diag}{\mathrm{diag}}

\newcommand*{\rank}{\mathrm{rank}}

\newcommand*{\id}{\mathrm{id}}

\newcommand*{\supp}{\mathrm{supp}}
\newcommand*{\tr}{\mathrm{tr}}

\newcommand{\proj}[1]{|#1\rangle\!\langle #1|}

\newcommand*{\di}{\mathrm{d}} 

\newcommand{\norm}[1]{\left\lVert#1\right\rVert}

\allowdisplaybreaks    

\begin{document}

\title{\LARGE Bounds on Lyapunov exponents via entropy accumulation}

\author[1]{David Sutter}
\author[2]{Omar Fawzi}
\author[1]{Renato Renner}

\affil[1]{\small{Institute for Theoretical Physics, ETH Zurich, Switzerland}}
\affil[2]{\small{Laboratoire de l'Informatique du Parall\'elisme, ENS de Lyon, France}}
\date{}

\maketitle

\begin{abstract}
Lyapunov exponents describe the asymptotic behavior of the singular values of large products of random matrices. A direct computation of these exponents is however often infeasible. By establishing a link between Lyapunov exponents and an information theoretic tool called entropy accumulation theorem we derive an upper and a lower bound for the maximal and minimal Lyapunov exponent, respectively. The bounds assume independence of the random matrices, are analytical, and are tight in the commutative case as well as in other scenarios. They can be expressed in terms of an optimization problem that only involves single matrices rather than large products. The upper bound for the maximal Lyapunov exponent can be evaluated efficiently via the theory of convex optimization.
\end{abstract}

\section{Introduction} \label{sec_intro}
Large products of random matrices arise in many areas of theoretical physics. Arguably the most important characterization of such products is given by the \emph{Lyapunov exponents} that describe the asymptotic behavior of the singular values. Oftentimes one encounters products that do not range over infinitely many matrices. This justifies the definition of non-asymptotic Lyapunov exponents that in the limit converge to the traditional Lyapunov exponents.

More precisely, let $n \in \N$ and $(L_{i})_{i \in [n]}$ be a sequence of random matrices on $\C^{d \times d}$ such that 
\begin{align} \label{eq_regAss}
\E \log^+ \! \sigma_{\max}(L_i) <\infty \quad \text{for all } i \in [n]:=\{1,2,\ldots,n\} \, ,
\end{align}
where $\sigma_{\max}$ denotes the largest singular value and $\log^+ t := \max\{\log t, 0 \}$. In this manuscript we consider sequences of random matrices only where assumption~\eqref{eq_regAss} holds. 
We define the \emph{non-asymptotic Lyapunov spectrum} of $(L_{i})_{i \in [n]}$, consisting of the \emph{non-asymptotic Lyapunov exponents}, as
\begin{align} \label{eq_defLyap_nonAsympt}
\gamma_{k,n} :=  \frac{1}{n} \E \log \sigma_k \Big(\prod_{i=1}^n L_i\Big) \quad \text{for} \quad 1 \leq k \leq d  \quad \text{and} \quad n \in \N \, ,
\end{align}
where $\sigma_k(\cdot)$ denotes the $k$-th singular value. We assume that the singular values are enumerated in decreasing order, i.e., $\sigma_{\max}=\sigma_1 \geq \sigma_2 \geq \ldots \geq  \sigma_d=\sigma_{\min}$.\footnote{In case two singular values are equal they are still assigned a different index.} Hence, the non-asymptotic Lyapunov spectrum is also ordered in the sense that $\infty > \gamma_{1,n} \geq \gamma_{2,n} \geq \ldots \geq \gamma_{d,n} \geq - \infty$ for all $n \in \N$. The expectation in~\eqref{eq_defLyap_nonAsympt} exists as the largest singular value is submultiplicative~\cite{bhatia_book} and by assumption~\eqref{eq_regAss}, but it can be $-\infty$.

For many applications it is natural to consider the traditional Lyapunov spectrum which is obtained by taking the limit $n \to \infty$. To ensure that the limit exists we need to impose further assumptions on the random matrices. Let $(L_{i})_{i \in \N}$ be a stationary sequence of random matrices that satisfies~\eqref{eq_regAss}. Then the (asymptotic) \emph{Lyapunov spectrum}, consisting of the (asymptotic) \emph{Lyapunov exponents}, is given by
\begin{align} \label{eq_defLyap}
\gamma_{k} := \lim_{n \to \infty} \gamma_{k,n} \quad \text{for} \quad 1 \leq k \leq d \, .
\end{align}
A precise argument of the well-known fact that the limit in~\eqref{eq_defLyap} exists is given in Appendix~\ref{app_limitExists} for the sake of completeness.

If we impose another assumption on the random matrices, i.e., that they are stationary and ergodic, it can be shown that $\PP$-almost surely the expectation in the definition of Lyapunov spectrum can be dropped, where $\PP$ is the probability measure of the underlying probability space the random matrices are defined on. More precisely, if $(L_{i})_{i \in \N}$ is a stationary ergodic sequence of random matrices on $\C^{d \times d}$ that satisfies~\eqref{eq_regAss} then
\begin{align} \label{eq_kingman}
\gamma_k = \lim_{n\to \infty} \frac{1}{n} \log \sigma_k \Big(\prod_{i=1}^n L_i\Big)  \quad \PP-\text{a.s.} \quad \text{for} \quad 1 \leq k \leq d  \, .
\end{align}
The justification for this is due to Kingman's subadditive ergodic theorem~\cite{kingman_73} (see also the two celebrated and historically older theorems by Oseledets~\cite{oseledets1968multiplicative} and Kesten-Furstenberg~\cite[Theorem~2]{furstenberg1960} which can be derived from Kingman's result).

The maximal Lyapunov exponent $\gamma_1$ has a dominant role within the Lyapunov spectrum. Its definition can be rewritten in terms of Schatten norms as
\begin{align} \label{eq_defMaxLyap}
\gamma_1 = \lim_{n\to \infty} \frac{1}{n} \E \log \norm{\prod_{i=1}^n L_i}_p  \quad \text{for} \quad p\geq 1\, ,
\end{align}
where $\norm{L}_p:=\big(\tr |L|^p\big)^{\frac{1}{p}}$ and $|L|:=\sqrt{L^\dagger L}$. In the limit $p\to \infty$ we recover the \emph{spectral norm} that is equal to the largest singular value denoted by $\sigma_{\max}(L)$ or $\sigma_1(L)$.
The maximal Lyapunov exponent defined in~\eqref{eq_defMaxLyap} is independent of the parameter $p$ since all Schatten $p$-norms are equivalent~\cite{bhatia_book} and $d< \infty$.

The minimal Lyapunov exponent $\gamma_d$ is a meaningful quantity only for distributions on the group of invertible matrices on $\C^{d \times d}$, denoted by $\GL(d,\C)$. This follows from the fact that whenever a probability distribution has positive weight on a matrix that is not invertible we have $\gamma_d = -Ê\infty$, which is obvious from the definition of $\gamma_d$.

The Lyapunov spectrum and in particular the maximal Lyapunov exponent plays a crucial role in several different areas of physics and mathematics. Arguably its most prominent applications are in the study of dynamical systems and of disordered materials. A positive maximal Lyapunov exponent for a sequence of random matrices describing the discretized time evolution of a dynamical system for example indicates that the system is chaotic. More precisely, the maximal Lyapunov exponent characterizes the sensitivity to initial conditions of a system. $\gamma_1$ is proportional to the inverse time rate at which two nearby trajectories diverge~\cite{brin2002introduction}. In the study of an Ising model, the maximal Lyapunov exponent is directly related to the free energy of the system and to the rate of the correlation decay~\cite{crisanti12}. In the Schr\"odinger equation with a random potential Lyapunov coefficients yield direct bounds on the localization length of the wave function~\cite{carmona_book,aizenman_book}. In information theory, the entropy rate of hidden Markov processes is directly related to the maximal Lyapunov exponent for a sequence of random matrices describing the Markov process~\cite{merhav02}.
Many more problems have been reduced to the study of the Lyapunov spectrum. We refer the interested reader to~\cite{crisanti12,aizenman_book} for more details.

Computing the Lyapunov spectrum and in particular the minimal and maximal Lyapunov exponents turns out to be challenging. There is no explicit formula known that can be evaluated easily.\footnote{Kingman mentions in~\cite{kingman_73}: \emph{``Pride of place among the unsolved problems of subadditive ergodic theory must go to the calculation of the constant $\gamma_1$.''}} 
An exception is the commutative case discussed in Section~\ref{sec_commutative}, i.e., if the random matrices commute pairwise. There are also a few specific sequences of $2\times 2$ matrices such that the corresponding maximal Lyapunov exponent can be computed (see, e.g., \cite{Mannion93,Marklof_08,Comtet2010,Comtet2013} and references therein). For arbitrary finite dimensions the maximal Lyapunov exponent is only known analytically for the case of independent and identically distributed (i.i.d.)~Gaussian matrices~\cite{Newman1986}, i.e., matrices with entries that are independent standard Gaussians, or small variations thereof~\cite{Forrester2013,Kargin2014}.
Furstenberg proved a powerful lower bound for the maximal Lyapunov exponent that plays an important role in the mathematical description of disordered materials. He showed that for matrices belonging to the special linear group $\SL(d,\R)$, i.e., the group of matrices over $\R^{d \times d}$ with determinant equal to $1$,  the maximal Lyapunov exponent is strictly positive for most distributions on this group~\cite{furstenberg1971}.

It has been shown that the maximal Lyapunov exponent cannot be approximated by an algorithm in full generality~\cite{tsit96,Tsitsiklis1997}. 
Under the assumption that the random matrices have nonnegative entries the problem is easier to deal with. For example certain limit type statements have been derived~\cite{hennion97,ramis19}. Furthermore, for random matrices with nonnegative entries there exist approximation algorithms for the maximal Lyapunov exponent that do converge~\cite{Pollicott2010,protasov13}, see also references therein for different algorithms with a slower rate of convergence. 

\paragraph{Result.} We prove that for any sequence $(L_i)_{i \in \N}$ of i.i.d.~random matrices on a semigroup $\cS \leq \C^{d \times d}$, or on a group $\cG \leq \GL(d,\C)$, that satisfies~\eqref{eq_regAss} we have
\begin{align} \label{eq_assBounds}
2\gamma_1 \leq \max_{X \in \X_{\cS}} \E \log \tr\,X L_1 L_1^{\dagger} \qquad \text{and} \quad 2 \gamma_d \geq \min_{X \in \X_{\cG}} \E \log \tr\,X L_1 L_1^{\dagger}\, ,
\end{align}
where $\X_{\cS}:=\{Y^\dagger Y/\tr\,Y^\dagger Y : Y \in \cS\}$. If we drop the (semi)group structure the optimizer can be assumed to be rank-one, i.e.,
\begin{align} \label{eq_assBounds2}
2\gamma_1 \leq \max_{X \in \X'_{\C^{d \times d}}} \E \log \tr\,X L_1 L_1^{\dagger} \qquad \text{and} \quad 2 \gamma_d \geq \min_{X \in \X'_{\C^{d \times d}}} \E \log \tr\,X L_1 L_1^{\dagger}\, ,
\end{align} 
where $\X'_{\C^{d \times d}}:=\{Y^\dagger Y/\tr\,Y^\dagger Y : Y \in \C^{d \times d} , \, \rank\, Y = 1\}$. We refer to Theorem~\ref{thm_main} for a more precise and more general result that provides upper and lower bounds for the non-asymptotic maximal and minimal Lyapunov exponents $\gamma_{1,n}$ and $\gamma_{d,n}$, respectively. The asymptotic statements are given by Corollaries~\ref{cor_main} and~\ref{cor2_main}.
The bounds are tight for independent and identically distributed diagonal matrices. This is discussed in Section~\ref{sec_commutative}.  Furthermore, Sections~\ref{sec_SL} and~\ref{sec_other_examples} present various scenarios where the bounds are either tight or outperform known bounds on the maximal or minimal Lyapunov exponent.  

The bounds~\eqref{eq_assBounds} are closely related with an information theoretic tool called \emph{entropy accumulation theorem} (EAT)~\cite{DFR16,DF18}. The EAT ensures that the operationally relevant entropic quantities of a multiparty system (called smooth min-and max-entropies~\cite{koenig09}) can be bounded by the sum of the von Neumann entropies of its individual parts viewed on a worst case scenario. 
In Section~\ref{sec_EAT} we explain this connection. In particular we show that the bounds~\eqref{eq_assBounds} have been motivated by the EAT and suggest an accumulation theorem for the relative entropy.

\paragraph{Efficient evaluation of the bounds.} One crucial difference between the bounds in~\eqref{eq_assBounds} for the maximal and minimal Lyapunov exponents and their definition~\eqref{eq_defLyap} is that the former are given by formulas that contain a single random matrix only and hence do not include limits of products of infinitely many random matrices. Depending on the structure of $\X_{\cS}$ and $\X_{\cG}$ the maximization and minimization in~\eqref{eq_assBounds} may still not be simple to evaluate. To circumvent this problem we can relax $\X_{\cS}$ and $\X_{\cG}$ to $\X_{\C^{d \times d}}$ which weakens the bounds but makes them easier to compute.
Proposition~\ref{prop_SDP} shows that the bound~\eqref{eq_assBounds} on the maximal Lyapunov exponent for $\X_{\C^{d \times d}}$ can be evaluated efficiently via \emph{convex programming}~\cite{boyd_book}.


\paragraph{Structure.} Section~\ref{sec_prelim} introduces the notation, reviews basic properties of eigenvalues of Hermitian matrices, and  summarizes known results on the continuity of the maximal Lyapunov exponent.
In Section~\ref{sec_main} we present and prove the main result and discuss its implications. Section~\ref{sec_examples} presents various examples that illustrate how to use the bounds of the main result in practice and give insights about their performance. Finally in Section~\ref{sec_connectionEntropy} we discuss two connections between Lyapunov exponents and entropy and explain why the bounds on $\gamma_1$ and $\gamma_d$ can be useful in this context.


\section{Preliminaries} \label{sec_prelim}
\subsection{Notation}
For $n \in \N$ let $[n]:=\{1,2,\ldots,n \}$.
We denote the L\"owner partial order on positive semidefinite matrices by $\geq$, i.e., $X\geq 0$ states that $X$ is a positive semidefinite matrix. For a matrix $L$ we write $L^\dagger$ for its conjugate transpose.
The natural logarithm is denoted by $\log (\cdot)$. The general and special linear group over $\C^{d \times d}$ are denoted by $\GL(d,\C)$ and $\SL(d,\C)$, respectively. If $\cG$ is a subgroup of $\cH$ we write $\cG \leq \cH$. For a semigroup $\cS \leq \C^{d \times d}$ we define the following set
\begin{align} \label{eq_mainSet}
\X_{\cS}:=\{Y^\dagger Y/\tr\,Y^\dagger Y : Y \in \cS\} \, .
\end{align}
If there exists $Y \in \cS$ such that $Y^\dagger Y=0$ then we define $Y^\dagger Y/\tr\,Y^\dagger Y=0$, i.e., $\X_{\cS}$ contains the zero matrix.\footnote{We note that this cannot happen if $\cS$ forms a group and hence is only relevant in case $\cS$ is a semigroup.}
We further define the set
\begin{align}
\X'_{\C^{d \times d}}:=\{Y^\dagger Y/\tr\,Y^\dagger Y : Y \in \C^{d \times d}, \, \rank \,Y =1\} \, .
\end{align}
Let $(L_n)_{n \in \N}$ be a sequence of random matrices on $\C^{d\times d}$ with $(\Omega, \cF, \PP)$ the associated probability space. Let $T:\Omega \to \Omega$ be a shift operator that drops the first coordinate and shifts the others one place to the left. An event $\cA \in \cF$ is said to be shift invariant if $\cA = T^{-1} \cA$.
$(L_n)_{n \in \N}$ is called \emph{stationary} if for every $k\geq 1$ it has the same distribution as the shifted sequence $(L_{k+n})_{n\in \N}$, i.e., for each $m$, $(L_1,\ldots, L_m)$ and $(L_k, . . . , L_{k+m})$ have the same distribution. 
$(L_n)_{n \in \N}$ is called \emph{ergodic} if every shift invariant event $\cA$ is trivial, i.e., $\PP(\cA)Ê\in \{0,1 \}$. We note that it is a simple exercise to show that i.i.d.~sequences are stationary and ergodic~\cite{durrett_book}.
\subsection{Variational formulas for eigenvalues}
For any matrix $L \in \C^{d\times d}$ we have $LL^\dagger \geq 0$ as $\langle z^\dagger L,L^\dagger z\rangle = \norm{L^\dagger z}^2 \geq 0$ for all $z \in \C^d$. The Cholesky decomposition ensures that every positive semidefinite matrix $0 \leq  A \in \C^{d \times d} $ can be written as $LL^\dagger$ for some $L\in \C^{d \times d}$~\cite[Fact~8.9.37]{bernstein_book}.\footnote{More precisely we can assume that $L$ is a lower triangular matrix with nonnegative diagonal entries.} The largest and smallest eigenvalue of $L L^\dagger$ are denoted by $\lambda_{\max}(L L^\dagger)$ and $\lambda_{\min}(L L^\dagger)$, respectively. These eigenvalues can be expressed as a semidefinite program. Using the notation~\eqref{eq_mainSet} we have 
\begin{align} \label{eq_smallestEig}
\lambda_{\max}(L L^\dagger) = \max_{X \in \X_{\C^{d \times d}}} \tr\, XL L^\dagger 
\qquad \text{and} \qquad
\lambda_{\min}(L L^\dagger) = \min_{X \in \X_{\C^{d \times d}}} \tr\, XL L^\dagger  \, .
\end{align}
For $L \in \cG \leq \GL(d,\C)$ the following relation holds
\begin{align} \label{eq_minEigGroup}
\lambda_{\min}(LL^\dagger) \geq  \frac{1}{d} \min_{Y \in \cG} \frac{ \tr\, Y^\dagger Y L L^\dagger }{\tr\, Y^\dagger Y}  \, .
\end{align}
This can be seen as for $Y=L^{-1} \in \cG$ the right-hand side simplifies to $(\tr\, (LL^{\dagger})^{-1})^{-1} = \|(LL^{\dagger})^{-1}\|_1^{-1} \leq \lambda_{\max}((LL^{\dagger})^{-1})^{-1}=\lambda_{\min}(L L^\dagger)$.
For $L \inÊ\C^{d\times d}$ we denote its singular values by $\sigma_{\max}(L)=\sigma_1(L) \geq \sigma_2(L) \geq \ldots \geq \sigma_d(L)=\sigma_{\min}(L)$. The largest and smallest singular value are related to the largest and smallest eigenvalues~\cite[Fact 9.13.1]{bernstein_book} by
\begin{align} \label{svd2eig}
\sigma_{\max}(L)^2 = \lambda_{\max}(L L^\dagger) \qquad \text{and}  \qquad \sigma_{\min}(L)^2 = \lambda_{\min}(L L^\dagger) \, .
\end{align}

\subsection{Continuity of the maximal Lyapunov exponent} \label{sec_continuity}
Consider a model where a fixed invertible matrix $L_i$ is chosen with probability $p_i>0$. The following theorem shows that the corresponding maximal Lyapunov exponent is continuous in $p_i$ and $L_i$.
\begin{theorem}[\cite{viana_prep}] \label{thm_viana}
The function $(p_i,L_i) \mapsto \gamma_1$ is continuous for $p_i > 0$  where $p_i$ are probabilities and $L_i$ invertible matrices. 
\end{theorem}
Because the reference~\cite{viana_prep} is in preparation we present a proof for the assertion of Theorem~\ref{thm_viana} in Appendix~\ref{app_details}, assuming the version of the theorem explicitly stated in~\cite[Theorem~3.5]{viana18}.
The following example discussed in~\cite{Kifer1982} shows that the continuity may break down if some of the probabilities vanish.
\begin{example}
Let $(L_i)_{i \in \N}$ be i.i.d.~matrices over
\begin{align}
\left \lbrace  A=\begin{pmatrix} \frac{1}{2} & 0 \\ 0 & 2 \end{pmatrix} , \,  B=\begin{pmatrix} 0 & -1 \\ 1 & 0 \end{pmatrix}  \right \rbrace \qquad \text{with probability} \quad \{p,1-p\} \quad  \text{for } p\in [0,1] \, .
\end{align}
It is straightforward to see that
\begin{align}
\gamma_1= \left \lbrace \begin{array}{l l}
\log 2 & \text{if } p=1 \\
0 & \text{if } p \in(0,1) 
\end{array} \right. \qquad
\text{and} \qquad \gamma_2= \left \lbrace \begin{array}{l l}
\log \frac{1}{2} & \text{if } p=1 \\
0 & \text{if } p \in(0,1)  \, .
\end{array} \right.
\end{align}
We note that this discontinuity of $\gamma_1$ is not in contradiction with the statement above since $\mu = p \delta_A + (1-p) \delta_B $ does not converge to $\delta_A$ in the $\cT$ topology when $p\to 1$, because $\supp \, \mu = \{A,BÊ\}$ does not converge to $\supp \, \delta = \{A \}$ in the Hausdorff topology.
\end{example}

\section{Main results and proofs} \label{sec_main}
We next state the main result which is an upper and lower bound for the non-asymptotic maximal and minimal Lyapunov exponents, respectively.
\begin{theorem} \label{thm_main}
Let $d,n \in \N$ and let $(L_i)_{i \in [n]}$ be a sequence of independent random matrices on a semigroup $\cS \leq \C^{d \times d}$ that satisfies~\eqref{eq_regAss}. Then
\begin{align} \label{eq_main_UB_general}
2\gamma_{1,n} 
\leq \frac{1}{n} \E \log \tr\, L_1 L_1^\dagger +  \frac{1}{n} \sum_{i=2}^n \max_{X \in \X_{\cS}} \E  \log \tr\,X L_i L_i^{\dagger}   \, ,
\end{align}
where $\X_{\cS}=\{Y^\dagger Y/\tr\,Y^\dagger Y : Y \in \cS\}$.
If $(L_i)_{i \in [n]}$ is distributed on a group $\cG \leq \mathrm{GL}(d,\C)$ we further have
\begin{align} \label{eq_main_LB_general}
2 \gamma_{d,n} \geq  \frac{1}{n} \sum_{i=1}^n \min_{X \in \X_{\cG}} \E  \log \tr\,X L_i L_i^{\dagger}   - \frac{\log d}{n}\, .
\end{align}
\end{theorem}
The proof of Theorem~\ref{thm_main} is given in Section~\ref{sec_pfMain} below.
\begin{remark}[Dependent random matrices with a discrete probability distribution.]
In case of discrete probability distributions on $\cS \leq \C^{d \times d}$ we can prove a version of Theorem~\ref{thm_main} without the independence assumption of the random matrices, i.e,
\begin{align} \label{eq_main_UB_general_dep}
2\gamma_{1,n} 
\leq \frac{1}{n} \E \log \tr\, L_1 L_1^\dagger +  \frac{1}{n} \sum_{i=2}^n \max_{X \in \X_{\cS}, \ell_1,\ldots,\ell_{i-1}  \in \cS} \E \big[ \log \tr\,X L_i L_i^{\dagger} | L_1= \ell_1, \ldots, L_{i-1}=\ell_{i-1}  \big]  \, ,
\end{align}
where $\X_{\cS}=\{Y^\dagger Y/\tr\,Y^\dagger Y : Y \in \cS\}$.
If $(L_i)_{i \in [n]}$ is distributed on a group $\cG \leq \mathrm{GL}(d,\C)$ we further have
\begin{align} \label{eq_main_LB_general_dep}
2 \gamma_{d,n} \geq  \frac{1}{n} \sum_{i=1}^n \min_{X \in \X_{\cG}, \ell_1,\ldots,\ell_{i-1}  \in \cG} \E \big[ \log \tr\,X L_i L_i^{\dagger} | L_1= \ell_1, \ldots, L_{i-1}=\ell_{i-1}  \big]  - \frac{\log d}{n}\, .
\end{align} 
We note that the proof follows the same lines as the proof of Theorem~\ref{thm_main} given in Section~\ref{sec_pfMain}.
The bounds~\eqref{eq_main_UB_general_dep} and~\eqref{eq_main_LB_general_dep} can be simplified if the random matrices satisfy a certain dependence structure. For example in case $(L_i)_{i \in [n]}$ form a Markov chain in order $L_{k-1} \leftrightarrow L_{k}\leftrightarrow L_{k+1}$ the formulas \eqref{eq_main_UB_general_dep} and~\eqref{eq_main_LB_general_dep} simplify to
\begin{align} \label{eq_main_UB_general_markov}
2\gamma_{1,n} 
\leq \frac{1}{n} \E \log \tr\, L_1 L_1^\dagger +  \frac{1}{n} \sum_{i=2}^n \max_{X \in \X_{\cS}, \ell_{i-1}  \in \cS} \E \big[ \log \tr\,X L_i L_i^{\dagger} | L_{i-1}=\ell_{i-1}  \big] 
\end{align}
and 
\begin{align} \label{eq_main_LB_general_markov}
2 \gamma_{d,n} \geq  \frac{1}{n} \sum_{i=1}^n \min_{X \in \X_{\cG}, \ell_{i-1}  \in \cG} \E \big[ \log \tr\,X L_i L_i^{\dagger} | L_{i-1}=\ell_{i-1}  \big]  - \frac{\log d}{n}\, .
\end{align}
\end{remark}

If the random matrices form a Markov chain or are even i.i.d.~the bounds can be further simplified in the limit $n \to \infty$.
\begin{corollary} \label{cor_main}
Let $d \in \N$ and let $(L_i)_{i \in \N}$ be a sequence of i.i.d.~random matrices on a semigroup $\cS \leq \C^{d \times d}$ that satisfies~\eqref{eq_regAss}. Then
\begin{align} \label{eq_main_UB}
2\gamma_1 \leq \max_{X \in \X_{\cS}} \E \log \tr\,X L_1 L_1^{\dagger} \ ,
\end{align}
where $\X_{\cS}=\{Y^\dagger Y/\tr\,Y^\dagger Y : Y \in \cS\}$.
If $(L_i)_{i \in \N}$ is distributed on a group $\cG \leq \mathrm{GL}(d,\C)$ we further have
\begin{align} \label{eq_main_LB}
2 \gamma_d \geq \min_{X \in \X_{\cG}} \E \log \tr\,X L_1 L_1^{\dagger}\, .
\end{align}
\end{corollary}
\begin{proof}
 The assertion follows from Theorem~\ref{thm_main} by considering the limit $n \to \infty$. The justification that the limits do exist for i.i.d.~random matrices is given in Section~\ref{sec_intro}. We note that assumption~\eqref{eq_regAss} implies that $\E \log \tr L_1 L_1^\dagger \leq d \, \E \log \lambda_{\max}(L_1 L_1^\dagger) < \infty$ and hence $\lim_{n \to \infty}  \frac{1}{n} \E \log \tr\, L_1 L_1^\dagger = 0$ which shows that the first term in~\eqref{eq_main_UB_general} vanishes in the limit $n \to \infty$. 
\end{proof}

There is another asymptotic version of Theorem~\ref{thm_main} where we enforce a rank-one constraint on the optimizers. 
\begin{corollary} \label{cor2_main}
Let $d \in \N$ and let $(L_i)_{i \in \N}$ be a sequence of i.i.d.~random matrices on $\GL(d,\C)$ with a distribution that satisfies~\eqref{eq_regAss} and has compact support. Then
\begin{align} \label{eq_main_UB_2}
2\gamma_1 \leq \max_{X \in \X'_{\C^{d \times d}}} \E \log \tr\,X L_1 L_1^{\dagger} \qquad \text{and} \qquad 2 \gamma_d \geq \min_{X \in \X'_{\C^{d \times d}}} \E \log \tr\,X L_1 L_1^{\dagger}\, ,
\end{align}
where $\X'_{\C^{d \times d}}=\{Y^\dagger Y/\tr\,Y^\dagger Y : Y \in \C^{d \times d}, \, \rank \,Y =1\}$.
\end{corollary}
The proof of Corollary~\ref{cor2_main} is given in Section~\ref{sec_proofCor}. We note that the major difference between Corollary~\ref{cor_main} and Corollary~\ref{cor2_main} is that in the latter we can assume that the optimizer has rank-one, at the cost of considering $\C^{d \times d}$ instead of $\cS$ or $\cG$.
Example~\ref{ex_rank1} shows that there exist scenarios where Corollary~\ref{cor2_main}  outperforms Corollary~\ref{cor_main}. Example~\ref{ex_rank1} presents a case where the bounds from Corollary~\ref{cor2_main} are even tight. 
\begin{remark} \label{rmk_generalization}
We note that Corollary~\ref{cor2_main} unlike Corollary~\ref{cor_main} does not respect the possible (semi)group structure of the random matrices. As visible from the proof it is possible to strengthen Corollary~\ref{cor2_main}. More precisely suppose we have a sequence of random matrices $(L_i)_{i \in \N}$ on a group $\cG \leq \C^{d \times d}$ with distribution $\mu$ such that there exits a family of joint distributions $(\mu'_n)_{n \in \N}$ such that
\begin{enumerate}
\item $\lim_{n \to \infty} \norm{\mu'_n - \mu}_1 = 0$ \label{it_1}
\item $\gamma_{1,n}' > \gamma_{2,n}'$ for all $n \in \N$ where $\gamma_{1,n}'$ and $\gamma_{2,n}'$ denote the largest and second largest Lyapunov exponents of $(L'_{i,n})_{i \in \N}$ distributed according to $\mu'_n$  \label{it_2}
\item for all $G \in \cG$ we have $\Pi_{GG^\dagger} \in \cG$, where $\Pi_{GG^\dagger}$ denotes the projector onto the eigenspace corresponding to the largest eigenvalue of $GG^\dagger$.  \label{it_3}
\end{enumerate}
Then it follows from the proof of Corollary~\ref{cor2_main} that~\eqref{eq_main_UB_2} is valid for $\X'_{\cG}$.

As an example, it is an easy exercise to verify that the group $\cG=\cD(d,\C)$ of commutative diagonalizable invertible matrices satisfies Properties~\ref{it_1},~\ref{it_2}, and~\ref{it_3} above with $\mu'_n = (1-\frac{1}{n}) \mu + \frac{1}{n} \tilde \mu$ where $\tilde \mu$ is the joint distribution of uniformly distributed matrices on $\cD(d,\C)$. We refer to Section~\ref{sec_commutative} for a more precise discussion of the commutative case.
\end{remark}

Evaluating the bounds in Corollary~\ref{cor_main} above may not be straightforward as the sets $\X_{\cS}$ and $\X_{\cG}$ can be complicated depending on the structure of $\cS$ and $\cG$. It is always possible to weaken the bounds by relaxing $\X_{\cS}$ and $\X_{\cG}$ to $\X_{\C^{d \times d}}$ which is equal to the set of positive semidefinite $d\times d$ matrices with trace one, also known as \emph{density matrices}.\footnote{We note that in case $\X_{\cS}$ contains the zero matrix this matrix can be removed from the set because for the upper bound the zero matrix would only be relevant in trivial scenarios where the upper bound is $-\infty$.} One important advantage of working with $\X_{\C^{d \times d}}$ is that due to the convexity of $\X_{\C^{d \times d}}$ the bound for $\gamma_1$ can be efficiently evaluated, i.e., the maximization is efficiently computable.  
\begin{proposition} \label{prop_SDP}
If the bound in~\eqref{eq_main_UB} is relaxed by using $\X_{\C^{d \times d}}$ instead of $\X_{\cS}$ it can be computed efficiently. More precisely, we have
\begin{align}
2\gamma_1 \leq \max_{X \in \X_{\C^{d\times d}}} \E \log \tr\,X L_1 L_1^{\dagger} \leq \log \max_{X \in \X_{\C^{d\times d}}}  \E \, \tr\, X L_1 L_1^{\dagger}\, ,
\end{align}
where the first and second bounds are a convex and semidefinite optimization problem, respectively.
\end{proposition}
\begin{proof}
We start by recalling that the set $\X_{\C^{d\times d}}$ can be written as
\begin{align}
\X_{\C^{d\times d}} = \{Y^\dagger Y/\tr\,Y^\dagger Y : Y \in \C^{d \times d}\} =\{ Y \in \C^{d \times d} :Y \geq 0, \tr\, Y =1\} \, . 
\end{align}
Jensen's inequality together with the monotonicity of the logarithm implies that
\begin{align}
\max_{X \in \X_{\C^{d\times d}}} \E \log \tr\,X L_1 L_1^{\dagger} \leq \log \max_{X \in \X_{\C^{d\times d}}}  \E \, \tr\, X L_1 L_1^{\dagger} \, .
\end{align}
The first expression is a convex optimization problem as we are maximizing a concave function over a convex set~\cite{boyd_book}. The second expression is even a semidefinite program as the objective function is linear.
\end{proof}
All semidefinite programs and most convex optimization problems can be solved efficiently by modern algorithms~\cite{boyd_book}.

The following remark compares the bounds from Theorem~\ref{thm_main} with bounds that can be obtained straightforwardly by using the submultiplicativity of the largest singular value. 
\begin{remark}[Comparison with trivial bounds] \label{rmk_simpleBounds}
The submultiplicativity of the largest singular value, i.e., $\sigma_{\max}(L_1 L_2) \leq \sigma_{\max}(L_1) \sigma_{\max}(L_2)$ for any two matrices $L_1$ and $L_2$~\cite{bhatia_book}, implies the following non-asymptotic bounds for the maximal and minimal Lyapunov exponent
\begin{align} \label{eq_simpleBound}
2\gamma_{1,n} \leq \frac{1}{n} \sum_{i=1}^n  \E\log\lambda_{\max}(L_i L_i^\dagger) \qquad \text{and} \qquad 2 \gamma_{d,n} \geq \frac{1}{n} \sum_{i=1}^n   \E \log\lambda_{\min}(L_i L_i^\dagger) \, ,
\end{align}
where for the lower bound we assume that $(L_n)_{n \in \N}$ is such that only invertible matrices have positive probability to occur (as otherwise $\gamma_d = - \infty$).
This is correct since
\begin{align} \label{eq_LyapunovUB}
2\gamma_{1,n} 
=  \frac{2}{n} \E \log  \sigma_{\max}\Big( \prod_{i=1}^n L_i \Big)
\leq  \frac{2}{n} \E  \sum_{i=1}^n \log \sigma_{\max}(L_i)
=\frac{1}{n} \sum_{i=1}^n \E \log \lambda_{\max}(L_i L_i^\dagger) \, , 
\end{align}
where the final step uses~\eqref{svd2eig}. The lower bound follows by similar arguments. By the submultiplicativity of the maximal singular value we have $\sigma_{\min}(L_1 L_2) \geq \sigma_{\min}(L_1) \sigma_{\min}(L_2)$ for any $L_1, L_2 \in \GL(d,\C)$.\footnote{This follows from that fact that for any invertible matrix $L$ we have $\sigma_{\max}(L)= 1/\sigma_{\min}(L^{-1})$.} Hence we obtain
\begin{align} \label{eq_LyapunovLB_trivial}
2\gamma_{d,n} 
=  \frac{2}{n} \E \log \sigma_{\min}\Big(\prod_{i=1}^n L_i\Big)
\geq  \frac{2}{n}  \E \sum_{i=1}^n \log \sigma_{\min}(L_i)
= \frac{1}{n} \sum_{i=1}^n  \E \log \lambda_{\min}(L_1 L_1^\dagger)\, ,
\end{align}
where the final step follows from~\eqref{svd2eig}.

If the random matrices are independent the bounds~\eqref{eq_simpleBound} immediately follow from~\eqref{eq_main_UB_general} and~\eqref{eq_main_LB_general}. To see this we rewrite~\eqref{eq_simpleBound} using the variational formulas~\eqref{eq_smallestEig} as
\begin{align} \label{eq_simple2}
2\gamma_{1,n} \leq  \frac{1}{n} \sum_{i=1}^n  \E  \max_{X \in \X_{\C^{d \times d}}} \log \tr\,X L_i L_i^{\dagger}  \qquad \text{and} \qquad 2 \gamma_{d,n} \geq   \frac{1}{n} \sum_{i=1}^n \E \min_{X \in \X_{\C^{d \times d}}} \log \tr\,X L_i L_i^{\dagger}  \ .
\end{align}
This shows that the bounds~\eqref{eq_main_UB_general} and~\eqref{eq_main_LB_general} are stronger than~\eqref{eq_simpleBound} and~\eqref{eq_simple2} as the former imply the latter by relaxing them (by using $\X_{\C^{d \times d}}$ instead of $\X_{\cS}$ and $\X_{\cG}$) and swapping the expectation with the maximization and minimization, respectively.\footnote{We note that in the lower bound for $\gamma_{d,n}$ in~\eqref{eq_main_LB_general} we have an additional term $O(\frac{\log d}{n})$ term, which however in practice does not matter as it vanishes for large values of $n$.}
We note that in case of dependent random matrices the bounds~\eqref{eq_simpleBound} may outperform~\eqref{eq_main_UB_general} and~\eqref{eq_main_LB_general}. 
\end{remark}

It is possible to distill further bounds on $\gamma_1$ and $\gamma_d$ from the bounds given by Theorem~\ref{thm_main} via a simple relation between the largest and smallest singular value of an invertible matrix, i.e., $\sigma_{\max}(L) = 1/\sigma_{\min}(L^{-1})$.
\begin{remark}
Let $(L_i)_{i \in \N}$ be a sequence of i.i.d.~random matrices on a group $\cG \leq \GL(d,\C)$. By utilizing the fact that $\sigma_{\max}(L) = 1/\sigma_{\min}(L^{-1})$ we find
 \begin{align}
 \gamma_1 
= \lim_{n\to \infty} \frac{1}{n} \E \log \sigma_{\max}Ê\Big( \prod_{i=1}^n L_i  \Big)
= - \lim_{n\to \infty} \frac{1}{n} \E \log \sigma_{\min}Ê\Big( \prod_{i=1}^n L_i^{-1}  \Big)
\end{align}
and an analogous expression for $\gamma_d$. We can apply this simple observation to the bounds from Corollary~\ref{cor_main} and find
\begin{align}
2 \gamma_1 \leq \min\big\{ \max_{X \in \X_{\cG}} \E \log \tr\,X L_1 L_1^{\dagger}, - \min_{X \in \X_{\cG}} \E \log \tr\,X (L_1^\dagger)^{-1} L_1^{-1} \big \}
\end{align}
and
\begin{align}
2 \gamma_d \geq \max \big\{\min_{X \in \X_{\cG}} \E \log \tr\,X L_1 L_1^{\dagger}, - \max_{X \in \X_{\cG}} \E \log \tr\,X(L_1^\dagger)^{-1} L_1^{-1} \big\} \, .
\end{align}
We note that instead of Corollary~\ref{cor_main} we could also improve the bounds of Theorem~\ref{thm_main} and Corollary~\ref{cor2_main} with this observation.
We note that in case the group $\cG$ has additional symmetry this may be useful to further improve the bounds. For example in case of $\cG=\SL(2,\R)$ we know that $\gamma_1 = - \gamma_2$ which turns out to be useful. This is explained in Example~\ref{ex_transferMatrix}.
\end{remark}

As mentioned in Section~\ref{sec_intro} for various applications it is of interest to show that the top Lyapunov exponent is strictly positive. Hence a good lower bound on $\gamma_1$ would be a desirable tool. The techniques used in this manuscript lead to an upper bound for $\gamma_1$ and a lower bound on $\gamma_d$, but presumably not to a good lower bound on $\gamma_1$.\footnote{A lower bound for $\gamma_d$ is by definition also a lower bound for $\gamma_1$, which however is not necessarily good since $\gamma_d$ and $\gamma_1$ can be far apart.} The interested reader may consult~\cite{lemm20} for recent advances on proving analytical lower bounds on $\gamma_1$.

We conclude this section with a remark about a generic possibility to further improve the bounds. The idea is to apply the new bounds to (small) products of matrices, i.e., for example we could consider a sequence of matrices $(L_{1,i}L_{2,i})_{i \in \N}$. The larger we choose these products, the better the bound performs at the cost that the evaluation of the bounds gets more complicated.

\subsection{Proof of Theorem~\ref{thm_main}} \label{sec_pfMain}
To simplify notation we denote $L^{n}:=\prod_{k=1}^n L_{k}$ and $L_m^{n}:=\prod_{k=m}^n L_{k}$ for $m \leq n$. Let $\mu_{L^n}$ denote the joint distribution of $L^n$.
We first prove the upper bound for the maximal Lyapunov exponent. By properties of Schatten norms, i.e., $\sigma_{\max}(A) \leq \norm{A}_2$ for every $ A \in \C^{d \times d}$ we have 
\begin{align}
2\gamma_{1,n} 
\leq  \frac{2}{n} \E \log \norm{L^n}_2
=  \frac{1}{n} \E \log \tr \, L^n (L^n)^\dagger  
\end{align}
and thus 
\begin{align}
2 \gamma_{1,n} 
\leq \frac{1}{n}\int  \log\big( \tr\, L^n (L^n)^{\dagger}\big)  \mu_{L^n}(\di L^n) \, .
\end{align}
We can write
\begin{align}
2 \gamma_{1,n}
&\leq \frac{1}{n} \int \log \tr\, L_1 L_1^\dagger \mu_{L_1}(\di L_1)  + \frac{1}{n}   \int  \sum_{i=2}^n \log\Big( \frac{\tr \, L^i (L^i)^\dagger}{\tr \, L^{i-1} (L^{i-1})^\dagger}\Big) \mu_{L^n}(\di L^n)\\
&=\frac{1}{n} \int \log \tr\, L_1 L_1^\dagger \mu_{L_1}(\di L_1) + \frac{1}{n}  \sum_{i=2}^n  \int  \log \Big( \frac{\tr \, L^i (L^i)^\dagger}{\tr \, L^{i-1} (L^{i-1})^\dagger}\Big)  \mu_{L^n}(\di L^n)  \\
&= \frac{1}{n} \int \log \tr\, L_1 L_1^\dagger \mu_{L_1}(\di L_1) \nonumber \\
&\hspace{10mm}+  \frac{1}{n} \sum_{i=2}^n   \int   \log \Big(\frac{\tr \, L^i (L^i)^\dagger}{\tr \, L^{i-1} (L^{i-1})^\dagger} \Big) \mu_{L_i}(\di L_i)  \mu_{(L^{i-1},L_{i+1}^n)}(\di L^{i-1}, \di L_{i+1}^n, L_i) \\
&=   \frac{1}{n} \E \log \tr L_1 L_1^\dagger + \frac{1}{n}   \sum_{i=2}^n   \int    \log \Big(\frac{\tr \, L^i (L^i)^\dagger}{\tr \, L^{i-1} (L^{i-1})^\dagger} \Big) \mu_{L_i}(\di L_i)  \mu_{(L^{i-1},L_{i+1}^n)}(\di L^{i-1}, \di L_{i+1}^n, L_i) \ ,
\end{align}
where $\mu_{(L^{i-1},L_{i+1}^n)}(\di L^{i-1}, \di L_{i+1}^n, L_i)$ denotes the conditional joint distribution of $(L^{i-1},L_{i+1}^n)$.
We now have a natural upper bound on each term $i \geq 2$
\begin{align}
\int   \log \Big(\frac{\tr \, L^i (L^i)^\dagger}{\tr \, L^{i-1} (L^{i-1})^\dagger} \Big) \mu_{L_i}(\di L_i)  \mu_{(L^{i-1},L_{i+1}^n)}(\di L^{i-1}\!\!, \di L_{i+1}^n, L_i) 
 \leq \max_{X \in \X_{\cS}} \E  \log \tr\,X L_i L_i^{\dagger}   \, . \label{eq_step1} 
\end{align}
We note that for $i \geq 2$ we have 
\begin{align} \label{eq_optimizerUB}
X= \frac{(L^{i-1})^\dagger L^{i-1}}{ \tr \, (L^{i-1})^\dagger L^{i-1}}  \in\X_{\cS} \, ,
\end{align}
which thus justifies the inequality above. 
Combining the previous steps gives
\begin{align} \label{eq_switchSets}
2 \gamma_{1,n} 
\leq \frac{1}{n} \E \log \tr L_1 L_1^\dagger + \frac{1}{n} \sum_{i=2}^n \max_{X \in \X_{\cS}} \E \log \tr\,X L_i L_i^{\dagger}    \, .
\end{align}

With a similar proof technique we obtain the asserted lower bound for the minimal Lyapunov exponent. By assumption $\mu$ is defined on $\cG \leq \GL(d,\C)$.
With the inequality for the smallest eigenvalue given in~\eqref{eq_minEigGroup} we have
\begin{align}
2 \gamma_{d,n} 
=  \frac{1}{n} \E\, \log \lambda_{\min}\big( L^n (L^n)^\dagger \big) 
\geq  \frac{1}{n} \E\, \log \min_{W  \in \cG}\frac{1}{\tr\, W^\dagger W} \tr\, W^\dagger W L^n  (L^n)^\dagger - \frac{\log d}{n}  \, .
\end{align}
Using the same notation as above we find
\begin{align}
2 \gamma_{d,n} 
\geq  \frac{1}{n} \int   \log \Big( \min_{W \in \cG} \frac{1}{\tr\, W^\dagger W} \tr\, W^\dagger W L^n (L^n)^\dagger \Big)  \mu_{L^n}(\di L^n) - \frac{\log d}{n} \, .
\end{align}
Let us denote the optimizer in the minimization above by $\bar W$.\footnote{Without loss of generality we can assume that $\tr \bar W^\dagger \bar W =1$ since we can renormalize the terms as we wish.} With some abuse of notation, i.e., $L^0 = \id_d$ we can write
\begin{align}
2 \gamma_{d,n} 
&\geq \frac{1}{n} \int  \sum_{i=1}^n \log \Big( \frac{\tr \, \bar W^\dagger \bar W L^{i} (L^{i})^{\dagger} }{\tr \, \bar W^\dagger \bar W L^{i-1} (L^{i-1})^{\dagger}  } \Big) \mu_{L^n}(\di L^n) - \frac{\log d}{n}  \label{eq_decS} \\
&=  \frac{1}{n}  \sum_{i=1}^n     \int   \log \Big( \frac{\tr \, \bar W^\dagger \bar W L^{i} (L^{i})^{\dagger} }{\tr \, \bar W^\dagger \bar W L^{i-1} (L^{i-1})^{\dagger}  } \Big)   \mu_{L^n}(\di L^n) - \frac{\log d}{n} \\
&=   \frac{1}{n} \sum_{i=1}^n \!  \int \!  \log \Big( \frac{\tr \, \bar W^\dagger \bar W L^{i} (L^{i})^{\dagger} }{\tr \, \bar W^\dagger \bar W L^{i-1} (L^{i-1})^{\dagger}  } \Big)  \mu_{L_i}(\di L_i)  \mu_{(L^{i-1},L_{i+1}^n)}(\di L^{i-1}, \di L_{i+1}^n, L_i)\! - \! \frac{\log d}{n} \ . \label{eq_decF}
\end{align}
Each term for $i\geq 1$ can be bounded from below as 
\begin{align}
&\int   \log \Big( \frac{\tr \, \bar W^\dagger \bar W L^{i} (L^{i})^{\dagger} }{\tr \, \bar W^\dagger \bar W L^{i-1} (L^{i-1})^{\dagger}  } \Big)  \mu_{L_i}(\di L_i)  \mu_{(L^{i-1},L_{i+1}^n)}(\di L^{i-1}, \di L_{i+1}^n, L_i) \geq \min_{X \in \X_{\cG}} \E \log \tr\,X L_i L_i^{\dagger} \, . \label{eq_stepD1}
\end{align}
The above inequality holds because
\begin{align} \label{eq_opt2}
X= \frac{ (L^{i-1})^\dagger \bar W^\dagger \bar W  L^{i-1}}{ \tr \, (L^{i-1})^\dagger \bar W^\dagger \bar W  L^{i-1}}  \in \X_{\cG} \, ,
\end{align}
which thus justifies the inequality above. Combining the previous steps gives
\begin{align} \label{eq_switchSets2}
2 \gamma_{d,n} 
\geq   \frac{1}{n} \sum_{i=1}^n \min_{X \in \X_{\cG}} \E  \log \tr\,X L_i L_i^{\dagger}  - \frac{\log d}{n} \, ,
\end{align}
which completes the proof.
\qed

\subsection{Proof of Corollary~\ref{cor2_main}} \label{sec_proofCor}
We note that the major difference between Theorem~\ref{thm_main} and Corollary~\ref{cor2_main} is that in the asymptotic setting of Corollary~\ref{cor2_main} we can assume that the optimizer is rank-one. Without this rank-one constraint the result would follow immediately from Theorem~\ref{thm_main} by considering the limit $n \to \infty$ which is explained in detail in the proof of Corollary~\ref{cor_main}. The justification that the limits do exist for stationary random matrices is given in Section~\ref{sec_intro}. 

It thus remains to prove why we can add the rank-one constraint to the optimizers. To see this recall that following the proof of Theorem~\ref{thm_main} we find
\begin{align}
2\gamma_1
\leq  \lim_{n \to \infty} \frac{1}{n} \E \log \tr\, L^n (L^n)^\dagger
= \lim_{n \to \infty} \frac{1}{n}  \sum_{i=2}^n  \int   \log \Big(\frac{\tr \, L^i (L^i)^\dagger}{\tr \, L^{i-1} (L^{i-1})^\dagger} \Big) \mu_{L^n}(\di L^n) \, . \label{eq_stepMid}
\end{align}
From Kingman's subadditive ergodic theorem, see~\eqref{eq_kingman}, we find that for any $\delta >0$ there is an $\eps_n$ with $\lim_{n \to \infty} \eps_n=0$ such that 
\begin{align} \label{eq_ASconv}
\PP\Big( \Big| \frac{1}{n} \log \lambda_{k} \big(L^n  (L^n)^\dagger \big) - 2 \gamma_k \Big | > \delta  \Big) \leq \eps_n \quad \text{for all }k \in[d] 
\end{align}
and hence
\begin{align} \label{eq_ASconv2}
\PP\Big( \lambda_{k} \big(L^n  (L^n)^\dagger \big) \ee^{ -2n \gamma_k}  > \ee^{n\delta}  \Big) \leq \eps_n \quad \text{and} \quad
\PP\Big( \lambda_{k} \big(L^n  (L^n)^\dagger \big) \ee^{ -2n \gamma_k}  < \ee^{-n\delta}  \Big) \leq \eps_n 
\quad \forall \, k \in[d] \, .
\end{align}
To simplify notation let
\begin{align}
M_n := (L^{n})^\dagger L^{n}  \qquad \text{and} \qquad X_n:=  \frac{M_n}{\tr\, M_n} \, . 
\end{align}
Consider the eigendecomposition $M_n= \sum_{k=1}^d \lambda_k^{(n)} \Pi_k^{(n)}$ where $\lambda_1^{(n)} \geq \lambda_2^{(n)} \geq \ldots \geq \lambda_d^{(n)}$ denote the eigenvalues of $M_n$ and $\Pi_k^{(n)}$ is the projector onto the eigenspace of $\lambda_k^{(n)}$. As a result we find
\begin{align}
X_n = \frac{\sum_{k=1}^d \lambda_k^{(n)} \Pi_k^{(n)}}{\sum_{j=1}^d \lambda_j^{(n)} } \, .
\end{align}
We note that $\lambda_k(L^n (L^n)^\dagger)= \lambda_k^{(n)}$ since since the eigenvalues of $Y Y^\dagger$ and $Y^\dagger Y$ are equal for any $Y \in \C^{d \times d}$. Statement~\eqref{eq_ASconv2} thus ensures  that for sufficiently large $n$ we have with high probability 
\begin{align}
\lambda_1^{(n)} \geq \ee^{n(2\gamma_1 - \delta)} \qquad \text{and} \qquad \lambda_k^{(n)} \leq \ee^{n(2\gamma_2 + \delta)} \quad \text{for all } k\geq 2 \, .
\end{align}
As a result we find with probability $1-\eps_n$
\begin{align}
\frac{\lambda_k^{(n)}}{\sum_{j=1}^d \lambda_j^{(n)}} \leq \ee^{-n (\gamma_1 - \gamma_2 + 2 \delta)} \quad \text{for all } k \geq 2 \, .
\end{align}

For the moment we assume that $\gamma_1 > \gamma_2$.
Plugging this into~\eqref{eq_stepMid} shows that 
\begin{align}
2\gamma_1
&= \lim_{n \to \infty} \frac{1}{n}  \sum_{i=2}^n  \int   \log \Big(\frac{\tr \, L^i (L^i)^\dagger}{\tr \, L^{i-1} (L^{i-1})^\dagger} \Big) \mu_{L^n}(\di L^n) \\
&= \lim_{n \to \infty} \frac{1}{n}  \sum_{i=2}^n  \int   \log( \tr\, X_{i-1} L_i L_i^\dagger)   \mu_{L^n}(\di L^n) \\
&= \lim_{n \to \infty} \frac{1}{n}  \sum_{i=2}^n  \int   \log \Big( \tr\, \frac{\sum_{k=1}^d \lambda_k^{(i-1)} \Pi_k^{(i-1)}}{\sum_{j=1}^d \lambda_j^{(i-1)}} L_i L_i^\dagger \Big)   \mu_{L^n}(\di L^n) \\
&\leq  \lim_{n \to \infty} \frac{1}{n}  \sum_{i=2}^n \int  \Big[ (1-\eps_n)    \log  \tr\, \big( (\Pi_1^{(i-1)} + \ee^{-(i-1)(\gamma_1-\gamma_2 +2 \delta)}  \id_d) L_i L_i^\dagger \big) \nonumber \\
&\hspace{30mm} + \eps_n  \log \Big(\frac{\tr \, L^i (L^i)^\dagger}{\tr \, L^{i-1} (L^{i-1})^\dagger} \Big) \Big]  \mu_{L_i}(\di L_i)  \mu_{(L^{i-1},L_{i+1}^n)}(\di L^{i-1}, \di L_{i+1}^n, L_i)\, .
\end{align}
Recall that by assumption $\E \log \sigma_{\max}(L_1) \leq \kappa <\infty$. 
If we split the sum in $i \in \{2,\ldots, \lfloor \sqrt{n} \rfloor\}$ and $i \in \{\lfloor \sqrt{n} \rfloor+1, \ldots, n \}$ we find 
\begin{align}
2\gamma_1
&\leq \lim_{n \to \infty} \frac{1}{n} \sum_{i=2}^{\lfloor \sqrt{n} \rfloor} \kappa + \lim_{n \to \infty} \frac{1}{n} \sum_{i=\lfloor \sqrt{n} \rfloor+1}^{n} \Big [ (1-\eps_n)  \int \log \tr \big( (\Pi_1^{(i-1)} + \ee^{-(i-1)(\gamma_1-\gamma_2 +2 \delta)} \id_d) L_i L_i^\dagger \big) \nonumber \\
&\hspace{60mm} \mu_{L_i}(\di L_i)  \mu_{(L^{i-1},L_{i+1}^n)}(\di L^{i-1}, \di L_{i+1}^n, L_i) + \eps_n \kappa \Big] \\
&\leq \lim_{n \to \infty} \frac{1-\eps_n}{n} \sum_{i=\lfloor \sqrt{n} \rfloor+1}^{n}   \int \log \tr \big( (\Pi_1^{(i-1)} + \ee^{-(\sqrt{n}-1)(\gamma_1-\gamma_2 +2 \delta)} \id_d) L_i L_i^\dagger \big) \nonumber \\
&\hspace{60mm}  \mu_{L_i}(\di L_i)  \mu_{(L^{i-1},L_{i+1}^n)}(\di L^{i-1}, \di L_{i+1}^n, L_i) \\
&= \lim_{n \to \infty}  \int \log \tr \, \Pi_1^{(n)} L_n L_n^\dagger \mu_{L_n}(\di L_n) \mu_{(L^{n-1})}(\di L^{n-1},L_n) \\
&\leq \lim_{n \to \infty}  \max_{X \in \X'_{\C^{d \times d}}}   \int \log \tr \, X L_n L_n^\dagger \, \mu_{L_n}(\di L_n) \\
&= \max_{X \in \X'_{\C^{d \times d}}} \E \log \tr\, X L_1 L_1^\dagger \, ,
\end{align}
for $\X'_{\C^{d \times d}}=\{Y^\dagger Y/\tr\,Y^\dagger Y : Y \in \C^{d \times d},\, \rank\, Y = 1\}$.
The second step uses that $\lim_{n \to \infty} \eps_n = 0$. The penultimate step is true since $\Pi_1^{(n)}$ is a rank-one projector for all $n \in \N$.

In case $\gamma_1 = \gamma_2$ we use a continuity argument to prove the assertion.
Let $\mu$ denote the joint distribution of the matrices $(L_i)_{i \in \N}$. We consider a family of joint distributions $(\mu'_n)_{n\in \N}$ on $\C^{d \times d}$ that is a perturbed version of $\mu$ of the form $\mu'_n = (1-\frac{1}{n}) \mu + \frac{1}{n} \tilde \mu$ where $\tilde \mu$ is the joint distribution of a sequence of matrices on $\C^{d \times d}$ whose entries are chosen uniformly at random with magnitude at most $1$. 
It is easy to see that $\lim_{n \to \infty} \norm{\mu'_n -\mu}_1=0$.
The Lyapunov spectrum of the sequence $(L'_{i,n})_{i \in \N}$ that is distributed according to $\mu'_n$ is simple (i.e., all Lyapunov exponents are distinct) for all $n \in \N$, i.e., in particular $\gamma_1(\mu'_n) > \gamma_2(\mu'_n)$ for all $n \in \N$. This follows from~\cite[Theorem~8.1]{viana_book} together with~\cite[Exercise~8.1]{viana_book}. 
Since $ \mu \mapsto \gamma_1(\mu)$ is continuous for this setup as explained in Section~\ref{sec_continuity}~\cite[Theorem~3.5]{viana18} we find
\begin{align}
2 \gamma_1(\mu) 
= 2 \lim_{n \to \infty} \gamma_1(\mu'_n) 
\leq  \lim_{n \to \infty}   \max_{X \in \X'_{\C^{d \times d}}} \E_{\mu'_n} \log \tr\, X L_1 L_1^\dagger \, . \label{eq_fistPartt}
\end{align}
H\"older's inequality implies that for all $X \in \X'$
\begin{align}
|\E_{\mu} \log \tr\, X L_1 L_1^\dagger - \E_{\mu'_n} \log \tr\, X L_1 L_1^\dagger| \leq \norm{\mu - \mu'_n}_1 \norm{\log \tr\, X L_1 L_1^\dagger }_{\infty} \, .
\end{align}
Since the matrix $L_1$ is invertible $\|\log \tr\, X L_1 L_1^\dagger \|_{\infty} = \kappa  < \infty$. Together with~\eqref{eq_fistPartt} this gives
\begin{align}
2 \gamma_1(\mu)
\leq  \lim_{n \to \infty} \max_{X \in \X'_{\C^{d\times d}}}  \E_{\mu} \log \tr\, X L_1 L_1^\dagger  + \kappa \norm{\mu'_n - \mu}_1
=  \max_{X \in \X'_{\C^{d\times d}}} \E_{\mu} \log \tr\, X L_1 L_1^\dagger \, .
\end{align}

It remains to prove the lower bound of $\gamma_d$ stated in Corollary~\ref{cor2_main}.
We note that $\X'_{\C^{d \times d}}$ is the boundary of $\X_{\C^{d \times d}}$. The concavity of the logarithm ensures the minimum is attained at the boundary which thus proves the assertion.
\qed

\subsection{Comparison of new bounds with existing results} \label{sec_comparison}
Since a direct calculation of Lyapunov exponents is known to be notoriously difficult, it is natural trying to derive good upper and lower bounds that can be evaluated efficiently. In Section~\ref{sec_intro} we discussed some existing results such as Furstenberg's celebrated lower bound on $\gamma_1$ for random matrices on the special linear group~\cite{furstenberg1971}. 

In this work (see Corollaries~\ref{cor_main} and~\ref{cor2_main}) we presented (i) an upper bound on $\gamma_1$ and (ii) a lower bound on $\gamma_d$. To the best of our knowledge the bound (ii) is entirely novel and we are not aware of a comparable existing bound in the literature. The bound (i) can be seen as a considerable refinement of a known result~\cite[Equation~(23)]{protasov13}. The novel bound has two important advantages
\begin{enumerate}
\item The upper bound from Corollary~\ref{cor_main} can utilize a possible semigroup structure of the random matrices. More precisely, the maximization in~\eqref{eq_main_UB} is over the set $\X_{\cS}$, whereas the existing bound~\cite[Equation~(23)]{protasov13} cannot reflect this structure and hence the maximization has to be done over the set $\X_{\C^{d \times d}}$. We note that $\X_{\cS}$ can be considerably smaller than $\X_{\C^{d \times d}}$ which makes the upper bound~\eqref{eq_main_UB} substantially better than~\cite[Equation~(23)]{protasov13}. This is illustrated by Example~\ref{ex_noncommutative1}. 
\item The upper bound from Corollary~\ref{cor2_main} is an improvement over~\cite[Equation~(23)]{protasov13} as the maximization is taken over the set $\X'_{\C^{d \times d}}$ instead of $\X_{\C^{d \times d}}$. Because $\X'_{\C^{d \times d}}$ is smaller than $\X_{\C^{d \times d}}$, as it requires the positive definite unit trace matrices to have rank one, the new bound can be substantially better compared to~\cite[Equation~(23)]{protasov13}. This is illustrated by Example~\ref{ex_rank1}.
\end{enumerate}


\section{Examples} \label{sec_examples}
In this section we discuss some examples and show how the bounds from Theorem~\ref{thm_main} and Corollaries~\ref{cor_main}~and~\ref{cor2_main} perform in practice. We start with a precise analysis of the commutative case.
\subsection{The commutative case} \label{sec_commutative}
Let $\cD(d,\C) \leq \GL(d,\C)$ denote the group of diagonal invertible matrices over $\C^{d \times d}$.
For any sequence $(L_i)_{i \in \N}$ of i.i.d.~random matrices on $\cD(d,\C)$ the maximal and minimal Lyapunov exponent are given by
\begin{align} \label{eq_classical}
2\gamma_{1} =  \lambda_{\max}( \E \, \log L_1 L_1^\dagger) \qquad \text{and} \qquad   2\gamma_{d} = \lambda_{\min}( \E \, \log L_1 L_1^\dagger)  \, ,
\end{align}
where $\lambda_{\max}(\cdot)$ and $\lambda_{\min}(\cdot)$ denote the largest and smallest eigenvalue, respectively.
This fact is formally proven in Section~\ref{sec_proofCommutative_formula} below and was already observed in~\cite{crisanti12}.
Let $\bar \cD(d,\C)$ denote the semigroup of diagonal (not necessarily invertible) matrices on $\C^{d \times d}$. Corollary~\ref{cor2_main} together with Remark~\ref{rmk_generalization} implies
\begin{align}
2 \gamma_{1} \leq   \max_{X \in \X'_{\bar \cD}}\E \log \tr\,X L_i L_i^{\dagger}  \qquad \text{and} \qquad  2 \gamma_d \geq  \min_{X \in \X'_{\bar \cD}}\E \log \tr\,X L_i L_i^{\dagger}   \, ,
\end{align}
for $ \X'_{\bar \cD} =  \{YY^\dagger / \tr\, YY^\dagger : Y \in \bar \cD(d,\C), \rank\, Y = 1\}$. As claimed in Section~\ref{sec_intro} our bounds are tight in the commutative case, i.e.,
\begin{align} \label{eq_clTight}
\lambda_{\max}( \E \, \log L_1 L_1^\dagger) =  \max_{X \in \X'_{\bar \cD}}\E \log \tr\,X L_1 L_1^{\dagger}
\end{align}
and
\begin{align} \label{eq_clTight2}
\lambda_{\min}( \E \, \log L_1 L_1^\dagger) =  \min_{X \in \X'_{\bar \cD}}\E \log \tr\,X L_1 L_1^{\dagger}\, .
\end{align}

\subsubsection{Proof of~\eqref{eq_clTight} and~\eqref{eq_clTight2}} \label{sec_ourBoundTight}
We start by proving~\eqref{eq_clTight}.
To see why this is correct we note that because $\bar \cD(d,\C)$ is the semigroup of commutative diagonalizable matrices there exists a unitary matrix $U$ that diagonalizes $X$ as well as $L L^\dagger$ for all $L \sim \mu$, i.e., $X = U \Lambda U^\dagger$ for $\Lambda = \diag(\lambda_1,\ldots,\lambda_d)$ with $\lambda_i = 1$ for some $i \in [d]$ and $\lambda_j=0$ for $j \ne i$, as well as $L L^\dagger = U \Xi U^\dagger$ for $\Xi = \diag(\xi_1, \ldots, \xi_d)$ with $\xi_i \geq 0$. We then find
\begin{align}
\max_{X \in \X'_{\bar \cD}} \E \log \tr\,X L L^{\dagger} 
= \max_{i \in [d]} \E\,  \log \xi_i 
= \lambda_{\max}(\E  \log  \Xi )
=\lambda_{\max}\big(U \E \log (\Xi)  U^\dagger \big)
=\lambda_{\max}( \E \, \log L L^\dagger) \,.
\end{align}
The same argumentation shows that 
\begin{align}
 \min_{X \in \X'_{\bar \cD}} \E \log \tr\,X L L^{\dagger} =\lambda_{\min}( \E \, \log L L^\dagger) 
\end{align}
is also correct.
\qed

\subsubsection{Proof of~\eqref{eq_classical}} \label{sec_proofCommutative_formula}
To prove~\eqref{eq_classical} we start with the following simple lemma.
\begin{lemma}\label{lem_spectralMapping}
	Let $H$ be a Hermitian matrix. Then
	\begin{align}
	\log\sigma_{\max}\big( \ee^{H}\big) = \lambda_{\max}(H) \qquad \text{and} \qquad \log\sigma_{\min}\big( \ee^{H}\big) = \lambda_{\min}(H)  \, .
	\end{align}
\end{lemma}
\begin{proof}
Every Hermitian matrix $H$ can be diagonalized, i.e., it can be written as $H=U \Lambda U^\dagger$, where $U$ is unitary and $\Lambda$ is a real diagonal matrix containing the eigenvalues of $H$. We thus find
	\begin{align}
	\log \sigma_{\max}\big(\ee^{H}\big) 
	= \log \sigma_{\max} \big(\ee^{U \Lambda U^{\dagger}}\big)
	= \log \sigma_{\max} \big(U \ee^{\Lambda} U^{\dagger}\big)
	= \log \sigma_{\max} \big(\ee^{\Lambda}\big)
	= \lambda_{\max}(H) \, .
	\end{align}
Analogously we have
	\begin{align}
	\log \sigma_{\min}(\ee^{H}) 
	= \log \sigma_{\min}(\ee^{U \Lambda U^{\dagger}})
	= \log \sigma_{\min}(U \ee^{\Lambda} U^{\dagger})
	= \log \sigma_{\min}(\ee^{\Lambda})
	= \lambda_{\min}(H) \, .
	\end{align}
\end{proof}
We prove the assertion for the maximal Lyapunov exponent.
Since $(L_{i})_{i \in\N}$ are distributed on $\cD(d,\C)$ we find 
	\begin{align}
	2\gamma_{1} 
	= \lim_{n \to \infty}  \frac{1}{n} \E \log  \sigma_{\max}\Big(\prod_{i=1}^nL_i\Big)^2
	&= \lim_{n \to \infty} \frac{1}{n} \E \log \sigma_{\max}\Big( \Big(\prod_{i=1}^nL_i\Big)\Big(\prod_{i=1}^nL_i\Big)^\dagger\Big) \\
	&= \lim_{n \to \infty} \frac{1}{n} \E \log \sigma_{\max}\Big(\prod_{i=1}^nL_i L_i^\dagger\Big)\\
		&= \lim_{n \to \infty} \frac{1}{n} \log \sigma_{\max}\Big(\prod_{i=1}^nL_i L_i^\dagger\Big)\quad \PP-\text{a.s.} \, ,
	\end{align}
	where the final step uses the Kesten-Furstenberg result.
	Lemma~\ref{lem_spectralMapping} and the fact that $(L_{i})_{i \in\N}$ are distributed on $\cD(d,\C)$ gives $\PP$-a.s.
\begin{align}
	2\gamma_{1} 
	=\lim_{n \to \infty}\frac{1}{n}  \log \sigma_{\max}\Big(\! \exp \Big(\sum_{i=1}^n \log L_i L_i^\dagger \Big)\! \Big)
	&= \lim_{n \to \infty}\frac{1}{n}  \lambda_{\max}\Big(\sum_{i=1}^n \log L_i L_i^\dagger\Big)\\
	&=\lambda_{\max}\Big( \lim_{n \to \infty} \frac{1}{n} \sum_{i=1}^n \log L_i L_i^\dagger\Big)\\
	&=\lambda_{\max}( \E\, L_1 L_1^\dagger) \, ,
\end{align}
where the penultimate step uses the continuity of the largest eigenvalue. The final step follows from the law of large numbers~\cite{durrett_book}.
The statement for the minimal Lyapunov exponent follows by the same line of arguments
\qed

\subsection{The special linear group} \label{sec_SL}
In this section we discuss one example that is of particular relevance in physics. The Lyapunov exponent can be utilized as a mathematical tool to understand properties of certain operators which explains the behavior of certain systems. The readers that are not familiar with this subject may directly jump to Example~\ref{ex_transferMatrix} keeping in mind that Lyapunov exponents for random matrices with determinant one, i.e., over the group $\SL(d,\R)$, have many applications in physics. For the more experienced reader we would like to give some further context why we consider Example~\ref{ex_transferMatrix} below.

Consider a random Schr\"odinger operator
\begin{align}
H_{\omega}= - \Delta+V_{\omega} \, ,
\end{align}
where $\Delta$ is the Laplacian and $V_{\omega}$ is a random potential. One goal is to identify typical spectral properties of such operators. In the one-dimensional case it is known that the operator $H_{\omega}$ has a complete set of eigenvectors that decay exponentially in space (see e.g.~\cite{aizenman_book}). To formally prove this statement the maximal Lyapunov exponent is useful. For an energy $E \in \R_+$ we can write the one-dimensional discrete Schr\"odinger equation as a difference equation of the form
\begin{align} \label{eq_differenceEq}
u_{n+1} + u_{n-1}+ V_{\omega,n} u_n = E u_n \, .
\end{align}
The maximal Lyapunov exponent for a given energy $E$, denoted by $\gamma_1(E)$ describes the exponential growth or decay of the solution to~\eqref{eq_differenceEq}. By iteration we find
\begin{align}
\begin{pmatrix} u_{n+1} \\ u_n   \end{pmatrix} =  \Big( \prod_{i=1}^n T_i \Big) \begin{pmatrix} u_{1} \\ u_0   \end{pmatrix} \qquad \text{with} \qquad T_i= \begin{pmatrix} E-V_{\omega,i} & -1 \\ 1 & 0   \end{pmatrix} \, .
\end{align}
The sequence of random matrices $(T_i)_{i \in \N}$ defined above are called \emph{transfer matrices} and the corresponding maximal Lyapunov exponent happens to be inverse proportional to the localization length. 
In the traditional study of random Schr\"odinger operators we use a lower bound on $\gamma_1$ (e.g., via the Furstenberg theorem~\cite{furstenberg1971}) to ensure that there exists a finite localization length.
On the other hand, an upper bound on $\gamma_1$ gives a lower bound on the localization length. In other words an upper bound on $\gamma_1$ gives an ultimate limit how small the localization length can be at most which is of general interest. There is a rich literature about properties of random operators. The interested reader can find more information about this subject in~\cite{carmona_book,aizenman_book}.

\begin{example} \label{ex_transferMatrix}
Let $(L_i)_{i \in \N}$ be i.i.d.~random matrices on $\SL(2,\R)$ of the form
\begin{align}
L_i=\begin{pmatrix} \omega_i & -1 \\ 1 & 0 \end{pmatrix} \, ,
\end{align}
where $(\omega_i)_{i \in \N}$ are i.i.d.~random variables uniformly distributed over $\{-1,1\}$. 
A well-known result by Furstenberg~\cite{furstenberg1971} implies that $\gamma_1 >0$. 
Using the $\SL(2,\R)$ group structure we find
\begin{align}
\X_{\cG}=\left \lbrace \frac{1}{a + c} \begin{pmatrix} a & \sqrt{ac-1} \\ \sqrt{ac-1} & c \end{pmatrix}, a,c \geq 0, ac \geq 1 \right \rbrace
\end{align}
and hence Corollary~\ref{cor_main} gives 
\begin{align}
 \gamma_1 
&\leq \frac{1}{4} \max_{a,c \geq 0, ac \geq 1} \left \lbrace \log \tr \frac{1}{a + c} \begin{pmatrix} a & \sqrt{ac-1} \\ \sqrt{ac-1} & c \end{pmatrix} \begin{pmatrix} 2 & -1 \\ -1 & 1 \end{pmatrix} \right. \nonumber \\
& \hspace{23mm} \left.  + \log \tr \frac{1}{a + c} \begin{pmatrix} a & \sqrt{ac-1} \\ \sqrt{ac-1} & c \end{pmatrix} \begin{pmatrix} 2 & 1 \\ 1 & 1 \end{pmatrix} \right \rbrace  \\
&=\frac{1}{4}\max_{a,c \geq 0 , ac \geq 1} \log \frac{4a^2 + c^2 +4}{(a+c)^2}
=\frac{1}{4}\log 4 
\approx 0.35\, .
\end{align}
Analogously we find
\begin{align}
\gamma_2 
\geq \frac{1}{4}\min_{a,c \geq 0 , ac \geq 1} \log \frac{4a^2 + c^2 +4}{(a+c)^2} 
= \frac{1}{4} \min_{a\geq 0} \log \frac{4 a^2+4}{5a^2+4} 
= \frac{1}{4} \log \frac{4}{5} 
\approx -0.06 \, , \label{eq_locEx}
\end{align}
where the second step uses that $c=\frac{4(a^2+1)}{a}$ is the minimizer. 
We can use the structure of $\SL(2,\R)$ to further improve the upper bound for $\gamma_1$. Note that $A \in \SL(2,\R)$ implies $\sigma_1(A)=\frac{1}{\sigma_2(A)}$. Hence, using~\eqref{eq_kingman} we find $\PP-$a.s.
\begin{align}
\gamma_1 
= \lim_{n \to \infty} \frac{1}{n} \log \sigma_1\left( \prod_{i=1}^n L_i \right)
=- \lim_{n \to \infty} \frac{1}{n} \log \sigma_2\left( \prod_{i=1}^n L_i \right)
=- \gamma_2
\leq - \frac{1}{4} \log \frac{4}{5}
\approx 0.06 \, ,
\end{align}
where we used~\eqref{eq_locEx} in the final step. Thus together with Furstenberg~\cite{furstenberg1971} we find $\gamma_2 \geq -0.06$ and $\gamma_1 \in (0, 0.06]$ which is a very accurate localization of the true value of $\gamma_1$.
As a comparison, the simple bounds from Remark~\ref{rmk_simpleBounds} give $\gamma_1 \leq \frac{1}{2}\log \frac{1}{2}(3+\sqrt{5})\approx 0.48$ and $\gamma_2 \geq \frac{1}{2} \log \frac{1}{2}(3-\sqrt{5})\approx -0.48$.
\end{example}

\subsection{Other interesting examples} \label{sec_other_examples}
In this section we discuss various other examples that illustrate how to use the bounds from Corollaries~\ref{cor_main} and~\ref{cor2_main} in practice.

\begin{example}[Rank-one matrices] \label{ex_noncommutative1}
Let $(L_i)_{i \in \N}$ be i.i.d.~matrices chosen uniformly over
\begin{align}
\left \lbrace \begin{pmatrix} 1 & 0 \\ 0 & 0 \end{pmatrix} , \,  \frac{1}{2}\begin{pmatrix} 1 & 1 \\ 1 & 1 \end{pmatrix}  \right \rbrace \, .
\end{align}
For this scenario we have $\gamma_1 = - \frac{1}{4} \log 2$ because each time in the large matrix product the matrix changes, which happens with probability $\frac{1}{2}$, we pick up a factor $2^{-1/2}$. Hence in the operator norm of the definition of the Lyapunov exponent we get a factor $2^{-n/4}$. To apply Corollary~\ref{cor_main} we first note that the distribution for the random matrices in this example is over a semigroup 
\begin{align}
\cS=\left \lbrace 2^{-k} \begin{pmatrix} 1 & 0 \\ 0 & 0 \end{pmatrix} ,  \frac{1}{2}\begin{pmatrix} 1 & 1 \\ 1 & 1 \end{pmatrix} ,   2^{-k} \begin{pmatrix} 1 & 1 \\ 0 & 0 \end{pmatrix}  \right \rbrace \,  \quad \text{for} \quad k \in \N
\end{align}
and hence 
\begin{align}
\X_{\cS}=\left \lbrace  \begin{pmatrix} 1 & 0 \\ 0 & 0 \end{pmatrix} ,  \frac{1}{2}\begin{pmatrix} 1 & 1 \\ 1 & 1 \end{pmatrix}  \right \rbrace \, .
\end{align}
Corollary~\ref{cor_main} then gives  $\gamma_1 \leq - \frac{1}{4} \log 2 \approx -0.1733$ which is tight for this example.
The simple bound from Remark~\ref{rmk_simpleBounds} gives $\gamma_1 \leq 0$. The known bound from~\cite[Equation~(23)]{protasov13} gives $\gamma_1\leq -0.0792$.
\end{example}

\begin{example}[Group structure] \label{ex_noncommutative2}
Let $(L_i)_{i \in \N}$ be i.i.d.~matrices chosen uniformly over
\begin{align}
\left \lbrace \begin{pmatrix} \sqrt{2} & 0 \\ 0 & \frac{1}{\sqrt{2}} \end{pmatrix} , \,  \begin{pmatrix} 0 & \frac{1}{\sqrt{3}} \\ -\sqrt{3} & 0 \end{pmatrix}  \right \rbrace \, .
\end{align}
For this scenario it has been shown that $\gamma_1 = 0$~\cite[Section II.6]{bougerol_book}.
Before applying the bounds from Corollary~\ref{cor_main} we note that the distribution for the random matrices in this example is over a group
\begin{align}
\cG=\left \lbrace  \begin{pmatrix} a & 0 \\ 0 & a^{-1} \end{pmatrix}, \begin{pmatrix} 0 & b^{-1} \\ -b & 0 \end{pmatrix}, a,b \in \R\backslash\{0\}  \right \rbrace\, 
\end{align}
and hence 
\begin{align}
\X_{\cG}=\left \lbrace \frac{1}{c+c^{-1}}  \begin{pmatrix} c & 0 \\ 0 & c^{-1} \end{pmatrix} : c \in \R_+ \right \rbrace \, .
\end{align}
Hence Corollary~\ref{cor_main} gives
\begin{align}
\gamma_1 
\leq \frac{1}{4} \max_{c \in \R_+} \left \lbrace \log \frac{9+c^2}{3+3c^2}  +\log \Big( 2- \frac{3}{2(1+c^2)} \Big) \right \rbrace 
= \frac{1}{2}\log\frac{35}{24}
\approx 0.1886 \, ,
\end{align}
where the maximizer is $c^*=\sqrt{\frac{19}{29}}$. Analogously we also find $\gamma_2 \geq \frac{1}{4} \log \frac{3}{4} \approx -0.07$.
The simple bounds from Remark~\ref{rmk_simpleBounds} give $\gamma_1 \leq \frac{1}{4} \log 6 \approx 0.45$ and $\gamma_2\geq -\frac{1}{4} \log 6 \approx -0.45$.
\end{example}

\begin{example}[Corollary~\ref{cor2_main} may outperform Corollary~\ref{cor_main}] \label{ex_rank1}
Let $(L_i)_{i \in \N}$ be i.i.d.~random matrices on $\C^{2\times 2}$ of the form
\begin{align}
L_i= U_i \begin{pmatrix} \alpha & 0 \\ 0 & \beta \end{pmatrix} U_i^{\dagger} \, ,
\end{align}
for $\alpha,\beta > 0$, $\alpha \ne \beta$, and $(U_i)_{i \in \N}$ randomly chosen according to the Haar measure on the unitary group $\mathrm{U}(2,\C)$.  Corollary~\ref{cor_main} (which in this case coincides with the bound from~\cite[Equation~(23)]{protasov13}) gives
\begin{align}
\gamma_1 
\leq \frac{1}{2} \max_{X \in \X_{\C^{d\times d}}} \E \log \tr\,X L_1 L_1^{\dagger}
=\frac{1}{2}  \E \log \frac{1}{2} \tr\, L_1 L_1^{\dagger}
=\frac{1}{2}  \log \frac{1}{2} (\alpha^2 + \beta^2) \, , \label{eq_worseUB}
\end{align}
where the second step uses that by symmetry the maximizer is $\frac{1}{2}\id_2$. Corollary~\ref{cor2_main} gives on the other hand gives
\begin{align}
\gamma_1 
\leq \frac{1}{2} \max_{X \in \X'_{\C^{d\times d}}} \E \log \tr\,X L_1 L_1^{\dagger}
= \frac{1}{2} \E \log  (1,0) L_1 L_1^{\dagger}(1,0)^\dagger \, , \label{eq_betterUB}
\end{align}
where we also used the symmetry of the Haar measure. The lower bound for $\gamma_2$ from Corollary~\ref{cor2_main} ensures that~\eqref{eq_betterUB} is actually tight since
\begin{align}
\gamma_2 
\geq  \frac{1}{2} \min_{X \in \X'_{\C^{d\times d}}} \E \log \tr\,X L_1 L_1^{\dagger}
= \frac{1}{2} \E \log  (1,0) L_1 L_1^{\dagger} (1,0)^{\dagger} \, ,
\end{align}
where we again used the symmetry of the Haar measure. Hence we can conclude that 
\begin{align}
\gamma_2 = \gamma_1 =  \frac{1}{2} \E \log  (1,0) L_1 L_1^{\dagger} (1,0)^{\dagger} \, .
\end{align}
Jensen's inequality (which is strict as $\alpha \ne \beta$) assures that~\eqref{eq_betterUB} is strictly better than~\eqref{eq_worseUB} since
\begin{align}
 \frac{1}{2} \E \log (1,0) L_1 L_1^{\dagger} (1,0)^{\dagger}
 &<  \frac{1}{2}  \log \E (1,0) L_1 L_1^{\dagger} (1,0)^{\dagger}\\
&\leq  \frac{1}{2}  \log \tr\, (\E \, U (1,0)^\dagger (1,0) U^\dagger  ) (\alpha^2 + \beta^2) \\
&=  \frac{1}{2}  \log \frac{1}{2}(\alpha^2 + \beta^2)  \, ,
\end{align}
where the penultimate step uses H\"older's inequality~\cite[Exercise IV.2.7]{bhatia_book}. The final step is true because $\E \, U (1,0)^\dagger (1,0) U^\dagger = \id_2/2$.
We note that depending on the value of $\alpha$ and $\beta$ the difference between~\eqref{eq_worseUB} and~\eqref{eq_betterUB} can be substantial. 
As an example for $\alpha=5$ and $\beta=1$ we obtain
\begin{align}
\gamma_1 \leq \frac{1}{2}  \log \frac{1}{2} (\alpha^2 + \beta^2) \approx 1.28
\qquad \text{and} \qquad
\gamma_1  = \gamma_2 = \frac{1}{2} \E \log  (1,0) L_1 L_1^{\dagger} (1,0)^\dagger \approx 1.18 \, .
\end{align}
The simple bounds from Remark~\ref{rmk_simpleBounds} give $\gamma_1 \leq \log \max\{\alpha,\beta\} \approx 1.61$ and $\gamma_2 \geq  \log \min\{\alpha,\beta\} =0$.
\end{example}

\begin{example}[Convex optimization solver to compute the bound] \label{ex_convexOptimization}
This example shows that in case of random matrices without any useful structure it is relevant that the upper bound from Proposition~\ref{prop_SDP} can be evaluated efficiently using convex programming.
 Let $(L_i)_{i \in \N}$ be i.i.d.~matrices chosen uniformly over
\begin{align*}
\left \lbrace  \!\!
\setlength\arraycolsep{2.2pt}
\begin{pmatrix*}[r]-5& 9& 6&-1& 5 \\ 1& 6& 5& 5& 2 \\ 6&-5& 5&-4& 1\\ 1&10&-9& 8& 2 \\  5&-4& 4&-8& 5 \end{pmatrix*}\!,\!   
\begin{pmatrix*}[r] 4&-6& 1& 2& 3 \\ 8& 7& 0& 1&-8 \\-8&-1& 4&10& 5\\ 0&-6 &-10&-7&6\\ 9&-8&5&-3&-10 \end{pmatrix*}\!,\!
\begin{pmatrix*}[r]6&-9&3&3&10    \\9&8&0&7&-10  \\ -1&2&-7&0&-6 \\  5&-10&-2&1&-1 \\ -4&10&2&-10&-5 \end{pmatrix*}\!,\!
\begin{pmatrix*}[r]3&9&-4&6&-2   \\ 0&9&4&-8&-9  \\ 5&3&3&-2&-9 \\ -8&-10&-7&6&-9\\-6&-8&-2&-1&-7 \end{pmatrix*} \!\!
 \right \rbrace \! .
\end{align*}
Proposition~\ref{prop_SDP} gives
\begin{align}
\gamma_1 
\leq  \frac{1}{2} \max_{X \in \X_{\C^{d\times d}}} \E \log \tr\,X L_1 L_1^{\dagger}
\approx 2.86 \, ,
\end{align}
where we used a convex optimization solver to compute the maximization. Duality theory of convex programming ensures the result is correct up to an error of order $O(10^{-9})$.\footnote{We used CVX on Matlab to solve the convex optimization problem. On a MacBook with 2.3 GHz Intel Core i7 and 8 GB memory we can run examples up to $d=500$ in a reasonable time (i.e., in a few minutes).} As a comparison, the simple bound from Remark~\ref{rmk_simpleBounds} gives $\gamma_1 \leq 3.05$.
\end{example}

\section{Connections between Lyapunov exponents and entropy} \label{sec_connectionEntropy}
In this section we discuss two connections between Lyapunov exponents and entropy. In particular we show that our bounds on $\gamma_1$ and $\gamma_d$ imply new relations for entropic quantities.

\subsection{Towards an entropy accumulation theorem for relative entropies} \label{sec_EAT}
Before we could prove the bounds on the maximal and minimal Lyapunov exponent we had to gain insight into what form they could have. A recent result from quantum information theory called \emph{entropy accumulation theorem}~\cite{DFR16,DF18} turned out to provide evidence on the structure of the bounds. On an informal level entropy accumulation ensures that the operationally relevant quantities of a multiparty system (called smooth min-and max-entropies~\cite{koenig09}) can be bounded by the sum of the von Neumann entropies of its individual parts viewed on a worst case scenario. Finally in the process of simplifying our argument we discovered the elementary proof for the main result that is presented in this manuscript. 

To make this connection more precise we show how a variant of the entropy accumulation theorem can be obtained as direct consequence of Theorem~\ref{thm_main}. More precisely, we show that
\begin{align} \label{eq_REAT}
 \frac{1}{n} D\big(\id_{A_1 \ldots A_n} \ \| \tr_R (\cM_n \circ \ldots \circ \cM_1)(\id_R) \big)
 \geq \frac{1}{n} \sum_{i=1}^n \min_{X \in \X_{\C^{d \times d}}}D\big(\id_{A_i} \| \tr_R \, \cM_i(X)\big) - o(1)\, ,
\end{align}
where $(\cM_i)_{i \in \N}$ are completely positive maps from $R \to A_i \otimes R$ defined by
\begin{align} \label{eq_maps}
\cM_i : X_R \mapsto  \int  \proj{\omega_i}_{A_i} \otimes L_i(\omega_i) X_R L_i(\omega_i)^\dagger  \mu_{L_i(\omega_i)} (\di L_i (\omega_i)) \, .
\end{align}
Inequality~\eqref{eq_REAT} may be viewed as a variant of the entropy accumulation theorem as it ensures that the relative entropy of a large system, represented by the left-hand side of~\eqref{eq_REAT} can be bounded as a sum of relative entropies of the individual systems.

In the following we show how~\eqref{eq_REAT} follows as a direct consequence from Theorem~\ref{thm_main}.
Consider a sequence $(L_i)_{i \in \N}$ of independent random matrices on $\C^d \simeq R$ with joint distribution $\mu_{L^n}$. By definition of the expectation value we have
\begin{align}
 \E \log \tr \, L_n \ldots L_1 L_1^\dagger \ldots L_n^\dagger
= \int  \log  \tr \, L_{n}(\omega_n) \ldots L_{1}(\omega_1) L_{1}(\omega_1)^\dagger \ldots L_{n}(\omega_n)^\dagger  \mu_{L^n(\omega^n)} (\di L^n (\omega^n)) \, .
\end{align}
Furthermore, by definition of the maps $(\cM_i)_{i \in \N}$ given in~\eqref{eq_maps} we find
that gives
\begin{align}
&(\cM_n \circ \ldots \circ \cM_1)(\id_R)_{A_1 \ldots A_n} \nonumber\\
&\hspace{10mm}=\int  \proj{\omega^n}_{A_1 \ldots A_n} \otimes \tr \, L_{n}(\omega_n) \ldots L_{1}(\omega_1) L_{1}(\omega_1)^\dagger \ldots L_{n}(\omega_n)^\dagger  \mu_{L^n(\omega^n)} (\di L^n (\omega^n)) \\
&\hspace{10mm}=:\tau_{A_1 \ldots A_n} \, .
\end{align}
Recalling that the quantum relative entropy is given by $D(\rho \| \sigma):=\tr\, \rho \log \rho - \tr\, \rho \log \sigma$ leads to
\begin{align}
 \frac{1}{n} D(\id_{A_1 \ldots A_n} \ \| \tau_{A_1 \ldots A_n})
=\! -  \frac{1}{n} \tr \log \tau_{A_1 \ldots A_n}
=\!- \frac{1}{n} \E \log \tr \, L_n \ldots L_1 L_1^\dagger \ldots L_n^\dagger
=\! - 2 \gamma_{1,n}- o(1)\, . \label{eq_reat1}
\end{align}
Furthermore, we find for all $i \in \N$
\begin{align}
\min_{X \in \X_{\C^{d \times d}}}D\big(\id_{A_i} \| \tr_R \, \cM_i(X)\big)
&=\!\min_{X \in \X_{\C^{d \times d}}}\!\! - \tr\, \log \int \tr(L_i(\omega_i) X L_i(\omega_i)^\dagger) \proj{\omega_i}_{A_i} \mu_{L_i(\omega_i)} (\di L_i (\omega_i)) \\
&= \min_{X \in \X_{\C^{d \times d}}} - \E \log \tr X L_i L_i^\dagger \\
&= - \max_{X \in \X_{\C^{d \times d}}}  \E \log \tr X L_i L_i^\dagger \, . \label{eq_reat2}
\end{align}
Theorem~\ref{thm_main} in combination with~\eqref{eq_reat1} and~\eqref{eq_reat2} thus implies~\eqref{eq_REAT}. 

This connection raises further questions such as the existence of a more general entropy accumulation result for relative entropies (where the first argument is not necessarily the identity operator) that contains the bounds on the Lyapunov exponents and the original entropy accumulation theorem for conditional entropies as a special case. Another step towards a fully general entropy accumulation theorem for relative entropies has been obtained recently in~\cite{sutter20} where a novel chain rule for the relative entropy is proven.

\subsection{Entropy rate of hidden Markov processes}\label{sec_applications}
In this section it is shown that the entropy rate of hidden Markov processes is directly related to the maximal Lyapunov exponent. Hence the results above can be used to obtain upper and lower bounds for the entropy rate.

Let $(X_i)_{i\in \N}$ be a stochastic stationary process. The entropy rate of this process is defined as
\begin{align}
\bar H(X):= \lim_{n \to \infty} \frac{1}{n} H(X^n)  
=-\lim_{n \to \infty} \frac{1}{n} \E \log P_{X^n}(X^n) \, ,
\end{align}
where the limit exists as the process is stationary. The interested reader may consult~\cite{cover} for more information about this quantity. The celebrated  Shannon-McMillan-Breiman (see, e.g.~\cite{cover_shannon}) theorem asserts that if $(X_n)_{\in \N}$ is also ergodic we have
\begin{align}
\bar H(X) = -\lim_{n \to \infty} \frac{1}{n} \log P_{X^n}(X^n) \quad \PP-\text{a.s.} \, .
\end{align}

Let $(X_n)_{n \in \N}$ be a stationary and ergodic Markov process taking values in a finite set $\cX$ described by a transition matrix $M \in [0,1]^{|\cX|\times |\cX|}$ such that $M_{x,x'}=\PP(X_{i+1}=x' | X_i =x)$. Let $(Y_n)_{n \in \N}$ denote a noisy version of the Markov process where the noise is described by a discrete memoryless channel.\footnote{A channel is said to be memoryless if the probability distribution of the output depends only on the input at that time and is conditionally independent of previous channel inputs or outputs.} A discrete channel consists of a discrete input alphabet $\cX$, a discrete output alphabet $\cY$, and a probability transition matrix $W \in [0,1]^{|\cX| \times |\cY|}$ such that  $W_{x,y}=\PP(Y_i=y|X_i=x)$.

The process $(Y_n)_{n \in \N}$ is a \emph{hidden Markov process}. These processes are well-studied and arise naturally in many areas of science ranging from statistics via communication and information theory to machine learning, just to name a few. The interested reader can find an extensive discussion about hidden Markov processes and their applications in~\cite{merhav02} and references therein.

Computing the entropy rate of a hidden Markov process is a complicated task and in general an explicit form is unknown. Interestingly the entropy rate of a hidden Markov process is closely related to the maximal Lyapunov exponent. A standard recursion~\cite{merhav02,weissman11} yields
\begin{align}
\PP(Y^n=y^n) = \mu^{\mathrm{T}} \left( \prod_{i=1}^n(M\odot W_{\cdot,y_i}^{\mathrm{T}}) \right) \mathds{1} \, ,
\end{align}
where $\mu$ is the stationary distribution of the Markov process $(X_n)_{n \in \N}$ (represented as a column vector), $M\odot W_{\cdot,y_i}^{\mathrm{T}}$ denotes the denotes the $|\cX|\times |\cX|$ matrix whose $x$th row is given by the componentwise multiplication of the $x$th row of $M$ by the row vector whose $x'$th component is $W_{x',y_i}$, and $\mathds{1}$ is the all-1 column vector. 
Since for a matrix $A$ with nonnegative entries $\mu^{\mathrm{T}} A \mathds{1}$ is a norm of $A$ and since all matrix norms are equivalent we have
\begin{align}
\bar H(Y) 
=-\lim_{n \to \infty}  \frac{1}{n} \E \log \norm{\prod_{i=1}^n (M\odot W_{\cdot,Y_i}^{\mathrm{T}})  } 
=-\gamma_1 \, .
\end{align}
In words, the entropy rate of the hidden Markov process $(Y_n)_{n \in \N}$ is equal to the negative maximal Lyapunov exponent of the random matrices $(M\odot W_{\cdot,Y_n}^{\mathrm{T}})_{nÊ\in \N}$. This connection was also observed and discussed in~\cite{goldsmith06,jacquet08}.


\paragraph{Acknowledgments.} We thank J\"urg Fr\"ohlich and Marius Lemm for discussions on random matrices and localization of random operators. We also thank Gian Michele Graf for making us aware of related literature, in particular~\cite[Equation~(23)]{protasov13}.
We further thank Lennart Baumg\"artner, Raban Iten, and Tobias Sutter for inspiring discussions about Lyapunov exponents~\cite{lennart17} and entropy rates of hidden Markov processes.
This work was funded by the Swiss National Science Foundation via project No.~200020\_165843 and via the National Centre of Competence in Research QSIT, by the Air Force Office of Scientific Research (AFOSR) via grant FA9550-19-1-0202, as well as by the French ANR project ANR-18-CE47-0011 (ACOM).

\appendix
\section*{Appendix}
\section{Lyapunov spectrum exists for stationary random matrices} \label{app_limitExists}
In this section we prove that for any sequence $(L_n)_{n \in \N}$ of stationary random matrices on $\C^{d \times d}$ such that $\E \log \sigma_{\max}(L_1) < \infty$ the limit in
\begin{align}  \label{eq_tsLimit}
\gamma_{k} := \lim_{n \to \infty} \gamma_{k,n} = \lim_{n \to \infty} \frac{1}{n} \E \log \sigma_k \Big(\prod_{i=1}^n L_i\Big)
\end{align}
exists for $1 \leq k \leq d$. We note that this fact is known and is explained here for completeness.

To prove this we need to introduce the \emph{antisymmetric tensor product}. For $j \in \N$ and a Hilbert space $\cH$ let $\cH^{\wedge j}$ denote the antisymmetric subspace of $\cH^{\otimes j}$. The $j$-th \emph{antisymmetric tensor power} $\wedge^j: \mathrm{L}(\cH) \to \mathrm{L}(\cH^{\wedge j})$ maps every matrix $L \in \cH$ to the restriction of $L^{\otimes j}$ to the antisymmetric subspace $\cH^{\wedge j}$ of $\cH^{\otimes i}$, where $\mathrm{L}(\cH)$ denotes the set of matrices on $\cH$. This mapping is well studied and oftentimes serves as a useful tool in proofs. Among other interesting properties~\cite[Section I.5 and p.~18]{bhatia_book} it satisfies for any $L_1, L_2 \in \C^{d\times d}$ and any $j \in \N$
\begin{align} \label{eq_wedge1}
\wedge^j(L_1 L_2)= (\wedge^{j} L_1) (\wedge^j L_2)
\end{align}
and
\begin{align} \label{eq_wedge2}
\sigma_{1}(\wedge^j L) = \prod_{i=1}^j \sigma_{i}(L) \, .
\end{align}

By property~\eqref{eq_wedge1} above and the submultiplicativity of the largest singular value~\cite{bhatia_book} we have for $n,m \in \N$
\begin{align}
a_{j,n+m}:= \E \log \sigma_{1}\Big( \wedge^{j} \prod_{i=1}^{n+m} L_i \Big)
&=\E \log \sigma_{1}\Big( \Big(\wedge^{j} \prod_{i=1}^{n} L_i\Big) \Big(\wedge^{j} \prod_{i=n+1}^{n+m} L_i\Big) \Big) \\
&\leq \E \log\sigma_{1} \Big( \wedge^j \prod_{i=1}^{n} L_i\Big) + \E \log\sigma_1 \Big( \wedge^j \prod_{i=n+1}^{m+n} L_i\Big) \\
&= \E \log\sigma_1 \Big( \wedge^j  \prod_{i=1}^{n} L_i\Big) + \E \log\sigma_1 \Big( \wedge^j   \prod_{i=1}^{m} L_i\Big)\\
&=a_{1,n} + a_{1,m} \, , 
\end{align}
where the penultimate step uses the assumption that $(L_n)_{n\in \N}$ are stationary. We thus see that $(a_{j,n})_{n\in \N}$ is a subadditive sequence and hence according to Fekete's subadditivity lemma~\cite{fekete23} the limit 
\begin{align}
\lim_{n \to \infty} \frac{1}{n} a_{j,n} = \inf_{n \in \N} \frac{1}{n}a_{j,n} = \xi_{j}
\end{align}
exists.

For $j=1$ we find with the help of~\eqref{eq_wedge2}
\begin{align}
\xi_1 = \lim_{n \to \infty} \frac{1}{n} a_{1,n}= \lim_{n \to \infty} \frac{1}{n} \E \log \sigma_1 \Big(\prod_{i=1}^n L_i\Big) = \lim_{n \to \infty} \gamma_{1,n} \, ,
\end{align}
i.e., the asserted limit in~\eqref{eq_tsLimit} exists for $k=1$.
For $j=2$ again using~\eqref{eq_wedge2} gives
\begin{align}
\xi_2 
= \lim_{n \to \infty} \frac{1}{n} a_{2,n} 
&= \lim_{n \to \infty} \frac{1}{n} \E \log \sigma_1 \Big(\prod_{i=1}^n L_i\Big) + \lim_{n \to \infty} \frac{1}{n} \E \log \sigma_2 \Big(\prod_{i=1}^n L_i\Big) \\ 
&= \xi_1 + \lim_{n \to \infty} \frac{1}{n} \E \log \sigma_2 \Big(\prod_{i=1}^n L_i\Big) \, ,
\end{align}
which shows that the asserted limit in~\eqref{eq_tsLimit} exists for $k=2$. We can now continue this argument to show that the limit in~\eqref{eq_tsLimit} exists for all $k \in [d]$. \qed

\section{Proof of Theorem~\ref{thm_viana}} \label{app_details}
Because the reference~\cite{viana_prep} is in preparation we present a proof for the assertion of Theorem~\ref{thm_viana} in Appendix~\ref{app_details} for completeness.

Let $C$ be a topological vector space and $C^*$ its continuous dual. The weak* topology on $C^*$ is defined to be the $C$-topology on $C^*$, i.e., the coarsest topology (the topology with the fewest open sets) under which every element $c \in C$ corresponds to a continuous map on $C^*$. 
We now take $C$ to be the space of continuous functions with the supremum norm on $\GL(d,\C)$. By the Riesz-Markov theorem its dual space $C^*$ can be identified with the space of all complex regular Borel measures of bounded variation on $\GL(d,\C)$ and the space of compactly supported probability measures on $\GL(d,\C)$, denoted by $\bar G(d,\C)$, can be identified with a subspace thereof. The weak* topology on $C^*$ can then be restricted to a topology on $\bar G(d,\C)$.
Let $\bar G(d,\C)$ be equipped with the weakest topology $\cT$ such that
\begin{enumerate}
\item $\cT$ is stronger than the weak* topology restricted to $\bar G(d,\C)$  \label{it_topo1}
\item $\cT$ is stronger than the pull-back of the Hausdorff topology by $\mu \mapsto \supp \, \mu$. \label{it_topo2}
\end{enumerate}
More information about these assumptions can be found in~\cite{viana18,folland2013}.


It has been shown~\cite[Theorem~3.5]{viana18} that $\mu \mapsto \gamma_1(\mu)$ is continuous on $\bar G(d,\C)$, where $\gamma_1$ denotes the maximal Lyapunov exponent.\footnote{We note that for $2\times 2$ matrices the continuity of the maximal Lyapunov exponent has been proven in~\cite{young_1986}.} This implies that for probability measures supported on finite sets, i.e.,
\begin{align} \label{eq_simple_meas}
\mu = \sum_{i=1}^n p_i \delta_{L_i}
\end{align}
the maximal Lyapunov exponent varies continuously with the probabilities $p_i>0$ and the matrices $L_i \in \GL(d,\C)$ at every point.
To see this it suffices to show that the function $(p,L) \mapsto \sum_{i=1}^n p_i \delta_{L_i}$ is continuous for $p=(p_1,\ldots,p_n)$ and $L=(L_1,\ldots,L_n)$, because $\mu \mapsto \gamma_1(\mu)$ is known to be continuous.  
Since $\cT$ is defined as the weakest topology satisfying points~\ref{it_topo1}.~and~\ref{it_topo2}., any open set of $\bar{G}(d,\C)$ is a union of an intersection between an open set according to the weak* topology and an open set according to the pull-back of the Hausdorff topology under the $\supp$ map.\footnote{To see why considering a union of finite intersections is sufficient, let $T'$ and $T''$ be two topologies on the same set. The weakest topology $T$ on that set that is stronger than $T'$ and $T''$ may be defined as the topology consisting of all open sets of the form 
\begin{align} \label{eq_topo}
Z = \cup_{i \in \cS} X_i \cap Y_i  \, ,
\end{align}
where $(X_i)$ and $(Y_i)$ are families of open sets of $T'$ and $T''$, respectively, parameterized by a (not necessarily countable) set $\cS$. 
Any (finite) intersection of sets $Z$ of the form~\eqref{eq_topo} is again of the form~\eqref{eq_topo}. The same is obviously true for (not necessarily finite) unions. So the set of all sets $Z$ of the form~\eqref{eq_topo} is a valid topology. Furthermore, any $X \in T'$ and any $Y \in T''$ can also be expressed in the form~\eqref{eq_topo}, as one can just take one $X_i$ to be equal to $X$ and $Y_i$ equal to the entire set, or analogously for $Y$. Hence $T$ is indeed at least as strong as $T'$ and $T''$. } 
It is therefore sufficient to prove that the pre-image of any such set under the map $(p, L) \mapsto \sum_{i=1}^n p_i \delta_{L_i}$ is open. This is done in two steps.

First, let $X$ be a subset of $\bar{G}(d,\C)$ that is open according to the weak* topology. This means that $X$ is a union of finite intersections of pre-images of open subsets of $\mathbb{R}$ under any map $\mu \mapsto \int \phi \mu$, where $\phi$ is a ``test function", i.e., a function in $C$. But this means that the pre-image of $X$ under~\eqref{eq_simple_meas} is a union of intersections of pre-images of open subsets of $\mathbb{R}$ under the concatenated map $(p, L) \mapsto \int [\phi \sum_i p_i \delta_{L_i}] = \sum_i p_i \phi(L_i)$. Since $\phi$ is continuous, this concatenated map is continuous, and hence its pre-images are indeed open (for any $p_i\geq 0$).

Second, let $X$ be a subset of $\bar{G}(d,\C)$ that is open according to the the pull-back of the Hausdorff topology under the supp map. By definition, this means that $X$ is itself the pre-image of an open set (on the set of subsets of $\GL(d,\C)$) under the supp map. But this means that the pre-image of $X$ is the pre-image of an open set (on the set of subsets of $\GL(d,\C)$) under the concatenated map $(p, L) \mapsto \textrm{supp}(\sum_i p_i \delta_{L_j}) = \{L_i : i \textnormal{ such that } p_i>0\}$. For $p_i>0$ this map is continuous according to the Hausdorff topology, and hence pre-images of open sets are open.

\bibliographystyle{arxiv_no_month}
\bibliography{bibliofile}

\begin{thebibliography}{10}

\bibitem{aizenman_book}
M.~Aizenman and S.~Warzel.
\newblock {\em Random Operators: Disorder Effects on Quantum Spectra and
  Dynamics}, volume 168.
\newblock American Mathematical Society, 2015.
\newblock
  \texttt{\href{http://dx.doi.org/10.1090/bull/1565}{DOI:\,10.1090/bull/1565}}.

\bibitem{cover_shannon}
P.~H. Algoet and T.~M. Cover.
\newblock A sandwich proof of the {S}hannon-{M}c{M}illan-{B}reiman theorem.
\newblock {\em The Annals of Probability}, 16(2):899--909, 1988.
\newblock Available online: \url{http://www.jstor.org/stable/2243846}.

\bibitem{viana_prep}
A.~Avila, A.~.Eskin, and M.~Viana.
\newblock Continuity of {L}yapunov exponents of random matrix products.
\newblock In preparation; personal communication with M.~Viana.

\bibitem{lennart17}
L.~Baumg\"artner.
\newblock Bounds for the maximal {L}yapunov exponent of random matrices, 2017.
\newblock Semester Thesis at ETH Zurich.

\bibitem{bernstein_book}
D.~S. Bernstein.
\newblock {\em Matrix {M}athematics}.
\newblock Princeton University Press, 2 edition, 2009.

\bibitem{bhatia_book}
R.~Bhatia.
\newblock {\em Matrix Analysis}.
\newblock Springer, 1997.
\newblock
  \texttt{\href{http://dx.doi.org/10.1007/978-1-4612-0653-8}{DOI:\,10.1007/978-1-4612-0653-8}}.

\bibitem{bougerol_book}
P.~Bougerol et~al.
\newblock {\em Products of random matrices with applications to
  {S}chr{\"o}dinger operators}, volume~8.
\newblock Springer Science \& Business Media, 2012.
\newblock
  \texttt{\href{http://dx.doi.org/10.1007/978-1-4684-9172-2}{DOI:\,10.1007/978-1-4684-9172-2}}.

\bibitem{boyd_book}
S.~Boyd and L.~Vandenberghe.
\newblock {\em Convex {O}ptimization}.
\newblock Cambridge University Press, 2004.
\newblock
  \texttt{\href{http://dx.doi.org/10.1017/CBO9780511804441}{DOI:\,10.1017/CBO9780511804441}}.

\bibitem{brin2002introduction}
M.~Brin and G.~Stuck.
\newblock {\em Introduction to Dynamical Systems}.
\newblock Cambridge University Press, 2002.
\newblock
  \texttt{\href{http://dx.doi.org/10.1017/CBO9780511755316}{DOI:\,10.1017/CBO9780511755316}}.

\bibitem{carmona_book}
R.~Carmona and J.~Lacroix.
\newblock {\em Spectral theory of random Schr{\"o}dinger operators}.
\newblock Birkh\"auser Boston, 1990.
\newblock
  \texttt{\href{http://dx.doi.org/10.1007/978-1-4612-4488-2}{DOI:\,10.1007/978-1-4612-4488-2}}.

\bibitem{Comtet2013}
A.~Comtet, J.-M. Luck, C.~Texier, and Y.~Tourigny.
\newblock The {L}yapunov exponent of products of random $2 \times 2$ matrices
  close to the identity.
\newblock {\em Journal of Statistical Physics}, 150(1):13--65, 2013.
\newblock
  \texttt{\href{http://dx.doi.org/10.1007/s10955-012-0674-8}{DOI:\,10.1007/s10955-012-0674-8}}.

\bibitem{Comtet2010}
A.~Comtet, C.~Texier, and Y.~Tourigny.
\newblock Products of random matrices and generalised quantum point scatterers.
\newblock {\em Journal of Statistical Physics}, 140(3):427--466, 2010.
\newblock
  \texttt{\href{http://dx.doi.org/10.1007/s10955-010-0005-x}{DOI:\,10.1007/s10955-010-0005-x}}.

\bibitem{cover}
T.~M. Cover and J.~A. Thomas.
\newblock {\em Elements of Information Theory}.
\newblock Wiley Interscience, 2006.
\newblock
  \texttt{\href{http://dx.doi.org/10.1002/047174882X}{DOI:\,10.1002/047174882X}}.

\bibitem{crisanti12}
A.~Crisanti, G.~Paladin, and A.~Vulpiani.
\newblock {\em Products of Random Matrices: in Statistical Physics}, volume
  104.
\newblock Springer, 2012.
\newblock
  \texttt{\href{http://dx.doi.org/10.1007/978-3-642-84942-8}{DOI:\,10.1007/978-3-642-84942-8}}.

\bibitem{DF18}
F.~{Dupuis} and O.~{Fawzi}.
\newblock Entropy accumulation with improved second-order term.
\newblock {\em IEEE Transactions on Information Theory}, 65(11):7596--7612,
  2019.
\newblock
  \texttt{\href{http://dx.doi.org/10.1109/TIT.2019.2929564}{DOI:\,10.1109/TIT.2019.2929564}}.

\bibitem{DFR16}
F.~Dupuis, O.~Fawzi, and R.~Renner.
\newblock Entropy accumulation, 2016.
\newblock \href{https://arxiv.org/pdf/1607.01796}{arXiv:1607.01796}.

\bibitem{durrett_book}
R.~Durrett.
\newblock {\em Probability: Theory and Examples}.
\newblock Cambridge University Press, 2010.
\newblock
  \texttt{\href{http://dx.doi.org/10.1017/CBO9780511779398}{DOI:\,10.1017/CBO9780511779398}}.

\bibitem{merhav02}
Y.~{Ephraim} and N.~{Merhav}.
\newblock Hidden {M}arkov processes.
\newblock {\em IEEE Transactions on Information Theory}, 48(6):1518--1569,
  2002.
\newblock
  \texttt{\href{http://dx.doi.org/10.1109/TIT.2002.1003838}{DOI:\,10.1109/TIT.2002.1003838}}.

\bibitem{sutter20}
K.~Fang, O.~Fawzi, R.~Renner, and D.~Sutter.
\newblock Chain rule for the quantum relative entropy.
\newblock {\em Phys. Rev. Lett.}, 124:100501, 2020.
\newblock
  \texttt{\href{http://dx.doi.org/10.1103/PhysRevLett.124.100501}{DOI:\,10.1103/PhysRevLett.124.100501}}.

\bibitem{fekete23}
M.~Fekete.
\newblock {\"U}ber die {V}erteilung der {W}urzeln bei gewissen algebraischen
  {G}leichungen mit ganzzahligen {K}oeffizienten.
\newblock {\em Mathematische Zeitschrift}, 17(1):228--249, 1923.

\bibitem{folland2013}
G.~B. Folland.
\newblock {\em Real analysis: modern techniques and their applications}.
\newblock John Wiley \& Sons, 2013.

\bibitem{Forrester2013}
P.~J. Forrester.
\newblock Lyapunov exponents for products of complex {G}aussian random
  matrices.
\newblock {\em Journal of Statistical Physics}, 151(5):796--808, 2013.
\newblock
  \texttt{\href{http://dx.doi.org/10.1007/s10955-013-0735-7}{DOI:\,10.1007/s10955-013-0735-7}}.

\bibitem{furstenberg1971}
H.~Furstenberg.
\newblock Random walks and discrete subgroups of {L}ie groups.
\newblock {\em Advances in probability and related topics}, 1:1--63, 1971.

\bibitem{furstenberg1960}
H.~Furstenberg and H.~Kesten.
\newblock Products of random matrices.
\newblock {\em Ann. Math. Statist.}, 31(2):457--469, 1960.
\newblock
  \texttt{\href{http://dx.doi.org/10.1214/aoms/1177705909}{DOI:\,10.1214/aoms/1177705909}}.

\bibitem{hennion97}
H.~Hennion.
\newblock Limit theorems for products of positive random matrices.
\newblock {\em The Annals of Probability}, 25(4):1545--1587, 1997.
\newblock Available online: \url{http://www.jstor.org/stable/2959507}.

\bibitem{goldsmith06}
T.~{Holliday}, A.~{Goldsmith}, and P.~{Glynn}.
\newblock Capacity of finite state channels based on {L}yapunov exponents of
  random matrices.
\newblock {\em IEEE Transactions on Information Theory}, 52(8):3509--3532,
  2006.
\newblock
  \texttt{\href{http://dx.doi.org/10.1109/TIT.2006.878230}{DOI:\,10.1109/TIT.2006.878230}}.

\bibitem{jacquet08}
P.~Jacquet, G.~Seroussi, and W.~Szpankowski.
\newblock On the entropy of a hidden {M}arkov process.
\newblock {\em Theoretical Computer Science}, 395(2):203 -- 219, 2008.
\newblock
  \texttt{\href{http://dx.doi.org/https://doi.org/10.1016/j.tcs.2008.01.012}{DOI:\,https://doi.org/10.1016/j.tcs.2008.01.012}}.

\bibitem{Kargin2014}
V.~Kargin.
\newblock On the largest {L}yapunov exponent for products of {G}aussian
  matrices.
\newblock {\em Journal of Statistical Physics}, 157(1):70--83, 2014.
\newblock
  \texttt{\href{http://dx.doi.org/10.1007/s10955-014-1077-9}{DOI:\,10.1007/s10955-014-1077-9}}.

\bibitem{Kifer1982}
Y.~Kifer.
\newblock Perturbations of random matrix products.
\newblock {\em Zeitschrift f{\"u}r Wahrscheinlichkeitstheorie und Verwandte
  Gebiete}, 61(1):83--95, 1982.
\newblock
  \texttt{\href{http://dx.doi.org/10.1007/BF00537227}{DOI:\,10.1007/BF00537227}}.

\bibitem{kingman_73}
J.~F.~C. Kingman.
\newblock Subadditive ergodic theory.
\newblock {\em The Annals of Probability}, 1(6):883--899, 1973.
\newblock Available online: \url{http://www.jstor.org/stable/2959077}.

\bibitem{koenig09}
R.~K\"onig, R.~Renner, and C.~Schaffner.
\newblock The operational meaning of min- and max-entropy.
\newblock 55(9):4337 --4347, 2009.
\newblock
  \texttt{\href{http://dx.doi.org/10.1109/TIT.2009.2025545}{DOI:\,10.1109/TIT.2009.2025545}}.

\bibitem{lemm20}
M.~Lemm and D.~Sutter.
\newblock Quantitative lower bounds on the {L}yapunov exponent from
  multivariate matrix inequalities, 2020.
\newblock \href{https://arxiv.org/abs/2001.09115}{arXiv:2001.09115}.

\bibitem{Mannion93}
D.~Mannion.
\newblock Products of $2 \times 2$ random matrices.
\newblock {\em The Annals of Applied Probability}, 3(4):1189--1218, 1993.
\newblock Available online: \url{http://www.jstor.org/stable/2245205}.

\bibitem{Marklof_08}
J.~Marklof, Y.~Tourigny, and L.~Wolowski.
\newblock Explicit invariant measures for products of random matrices.
\newblock {\em Trans. Amer. Math. Soc}, 2008.
\newblock
  \texttt{\href{http://dx.doi.org/https://doi.org/10.1090/S0002-9947-08-04316-X}{DOI:\,https://doi.org/10.1090/S0002-9947-08-04316-X}}.

\bibitem{ramis19}
R.~Movassagh and J.~Schenker.
\newblock An ergodic theorem for homogeneously distributed quantum channels
  with applications to matrix product states, 2019.
\newblock \href{https://arxiv.org/abs/1909.11769}{arXiv:1909.11769}.

\bibitem{Newman1986}
C.~M. Newman.
\newblock The distribution of {L}yapunov exponents: Exact results for random
  matrices.
\newblock {\em Communications in Mathematical Physics}, 103(1):121--126, 1986.
\newblock
  \texttt{\href{http://dx.doi.org/10.1007/BF01464284}{DOI:\,10.1007/BF01464284}}.

\bibitem{weissman11}
E.~Ordentlich and T.~Weissman.
\newblock Bounds on the entropy rate of binary hidden {M}arkov processes.
\newblock {\em Cambridge University Press}, 4, 2011.
\newblock
  \texttt{\href{http://dx.doi.org/10.1017/CBO9780511819407.005}{DOI:\,10.1017/CBO9780511819407.005}}.

\bibitem{oseledets1968multiplicative}
V.~I. Oseledets.
\newblock A multiplicative ergodic theorem. characteristic {L}yapunov,
  exponents of dynamical systems.
\newblock {\em Trudy Moskovskogo Matematicheskogo Obshchestva}, 19:179--210,
  1968.

\bibitem{Pollicott2010}
M.~Pollicott.
\newblock Maximal {L}yapunov exponents for random matrix products.
\newblock {\em Inventiones mathematicae}, 181(1):209--226, 2010.
\newblock
  \texttt{\href{http://dx.doi.org/10.1007/s00222-010-0246-y}{DOI:\,10.1007/s00222-010-0246-y}}.

\bibitem{protasov13}
V.~Protasov and R.~Jungers.
\newblock Lower and upper bounds for the largest {L}yapunov exponent of
  matrices.
\newblock {\em Linear Algebra and its Applications}, 438(11):4448 -- 4468,
  2013.
\newblock
  \texttt{\href{http://dx.doi.org/10.1016/j.laa.2013.01.027}{DOI:\,10.1016/j.laa.2013.01.027}}.

\bibitem{tsit96}
J.~N. Tsitsiklis and V.~D. Blondel.
\newblock The spectral radius of a pair of matrices is hard to compute.
\newblock In {\em Proceedings of 35th IEEE Conference on Decision and Control},
  volume~3, pages 3192--3197 vol.3, 1996.
\newblock
  \texttt{\href{http://dx.doi.org/10.1109/CDC.1996.573624}{DOI:\,10.1109/CDC.1996.573624}}.

\bibitem{Tsitsiklis1997}
J.~N. Tsitsiklis and V.~D. Blondel.
\newblock The {L}yapunov exponent and joint spectral radius of pairs of
  matrices are hard---when not impossible---to compute and to approximate.
\newblock {\em Mathematics of Control, Signals and Systems}, 10(1):31--40,
  1997.
\newblock
  \texttt{\href{http://dx.doi.org/10.1007/BF01219774}{DOI:\,10.1007/BF01219774}}.

\bibitem{viana_book}
M.~Viana.
\newblock {\em Lectures on Lyapunov Exponents}.
\newblock Cambridge University Press, 2014.
\newblock
  \texttt{\href{http://dx.doi.org/10.1017/CBO9781139976602}{DOI:\,10.1017/CBO9781139976602}}.

\bibitem{viana18}
M.~Viana.
\newblock ({D}is)continuity of {L}yapunov exponents.
\newblock {\em Ergodic Theory and Dynamical Systems}, 2018.
\newblock
  \texttt{\href{http://dx.doi.org/10.1017/etds.2018.50}{DOI:\,10.1017/etds.2018.50}}.

\bibitem{young_1986}
L.-S. Young.
\newblock Random perturbations of matrix cocycles.
\newblock {\em Ergodic Theory and Dynamical Systems}, 6(4):627Ð637, 1986.
\newblock
  \texttt{\href{http://dx.doi.org/10.1017/S0143385700003734}{DOI:\,10.1017/S0143385700003734}}.

\end{thebibliography}
      
\end{document}